\documentclass[journal,10pt]{IEEEtran}

\usepackage{ifpdf}
\usepackage{cite}
\usepackage{amsmath}
\usepackage{array}
\usepackage{fixltx2e}
\usepackage{stfloats}
\usepackage{url}

\usepackage{graphicx,psfrag,epsf}
\usepackage{mathtools}
\usepackage{booktabs}
\usepackage[english]{babel} 
\usepackage{amsmath,amsfonts,amsthm}
\usepackage{chngcntr}
\usepackage{ amssymb }
\usepackage{array}
\usepackage{booktabs} 
\usepackage{setspace}
\usepackage{fix-cm}
\usepackage{xcolor}
\usepackage[normalem]{ulem}
\usepackage{hyperref}
\hypersetup{colorlinks,urlcolor=blue}
\usepackage{url}
\usepackage{tikz}
\usepackage{cool}
\usepackage{enumitem}
\usepackage{soul}
\usepackage{threeparttable}
\usepackage{tablefootnote}
\usepackage{color}
\usepackage{tabularx}
\usepackage{footnote}
\usepackage{algorithm}
\usepackage[noend]{algpseudocode}
\usepackage{footnote}
\usepackage{multirow}

\makeatletter
\DeclareUrlCommand\ULurl@@{%
  \def\UrlLeft{\bgroup}%
  \def\UrlRight{\egroup}}
\def\ULurl@#1{\hyper@linkurl{\ULurl@@{#1}}{#1}}
\DeclareRobustCommand*\ULurl{\hyper@normalise\ULurl@}

\newcommand{\distas}[1]{\mathbin{\overset{#1}{\kern\z@\sim}}}%
\newcommand{\bm}[1]{\mathbf{#1}}
\newsavebox{\mybox}\newsavebox{\mysim}
\newcommand{\distras}[1]{%
  \savebox{\mybox}{\hbox{\kern3pt$\scriptstyle#1$\kern3pt}}%
  \savebox{\mysim}{\hbox{$\sim$}}%
  \mathbin{\overset{#1}{\kern\z@\resizebox{\wd\mybox}{\ht\mysim}{$\sim$}}}%
}
\newtheorem{theorem}{Theorem}

\newtheorem{definition}{Definition}

\newtheorem{lemma}[theorem]{Lemma}

\newcolumntype{C}[1]{>{\centering\let\newline\\\arraybackslash\hspace{0pt}}m{#1}}

\newcommand{\be}{\begin{equation}}
\newcommand{\ee}{\end{equation}}
\newcommand{\bi}{\begin{itemize}}
\newcommand{\ei}{\end{itemize}}
\newcommand{\ben}{\begin{enumerate}}
\newcommand{\een}{\end{enumerate}}

\newcommand{\cbl}[1]{{\color{black}{#1}}}
\allowdisplaybreaks

\DeclareMathOperator*{\argmax}{\arg\!\max}
\makeatother

\let\oldbibliography\thebibliography
\renewcommand{\thebibliography}[1]{\oldbibliography{#1}\setlength{\itemsep}{0pt}}

\begin{document}
\title{Maximum entropy low-rank matrix recovery}
\author{Simon~Mak, ~
        Yao~Xie
\thanks{Simon Mak (e-mail: \ULurl{smak6@gatech.edu}) and Yao Xie (e-mail: \ULurl{yao.xie@isye.gatech.edu}) are with the H. Milton Stewart School of Industrial and Systems Engineering, Georgia Institute of Technology, Atlanta, GA.}
\thanks{This work was partially supported by NSF grants CCF-1442635, CMMI-1538746, and an NSF CAREER
Award CCF-1650913.}
}



\maketitle

\begin{abstract}
We propose in this paper a novel, information-theoretic method, called \texttt{MaxEnt}, for efficient data acquisition for low-rank matrix recovery. This proposed method has important applications to a wide range of problems, including image processing and text document indexing. Fundamental to our design approach is the so-called maximum entropy principle, which states that the measurement masks which maximize the entropy of observations, also maximize the information gain on the unknown matrix $\bm{X}$. Coupled with a low-rank stochastic model for $\bm{X}$, such a principle (i) reveals novel connections between information-theoretic sampling and {subspace packings}, and (ii) yields efficient mask construction algorithms for matrix recovery, which significantly outperforms random measurements. We illustrate the effectiveness of \texttt{MaxEnt} in simulation experiments, and demonstrate its usefulness in two real-world applications on image recovery and text document indexing.
\end{abstract}


%
\IEEEpeerreviewmaketitle

\vspace{-0.1in}
\section{Introduction}

Low-rank matrices play a fundamental role in solving a wide range of statistical, engineering and machine learning problems. For many such problems, however, the low-rank matrix $\bm{X} \in \mathbb{R}^{m_1 \times m_2}$ cannot be directly or fully observed as data. Instead, the observations $\bm{y} \in \mathbb{R}^{n}$ are typically obtained as $\bm{y} = \mathcal{A}(\bm{X}) + \boldsymbol{\epsilon}$, where $\mathcal{A}:\mathbb{R}^{m_1 \times m_2} \rightarrow \mathbb{R}^n$ is a linear measurement operator, and $\boldsymbol{\epsilon}$ is a vector of measurement noise. The goal then is to recover the underlying matrix $\bm{X}$ from observations $\bm{y}$, a problem known as \textit{matrix recovery}. For many applications in the physical or biological sciences, a key challenge is the cost of obtaining measurements from $\bm{X}$, which can be quite expensive. One example is in gene studies \cite{ND2014}, where costly, time-intensive experiments are needed to observe the matrix $\bm{X}$ of gene-disease expression levels. This is further compounded for high-dimensional matrices, where $m_1$, $m_2$ and $n$ are large. In light of this challenge, we propose in this paper a novel, information-theoretic method for \textit{designing} the measurement operator $\mathcal{A}$, so that more information  can be extracted and a better recovery can be achieved on $\bm{X}$ compared to random measurements.

Given the increasing prevalence of low-rank modeling in scientific and engineering problems \cite{DR2016}, the proposed methodology has important applications to a broad spectrum of important problems. We briefly review two such problems which motivated this work:
\bi
\item \textit{Image processing}: In many imaging systems, the measurements $\bm{y}$ are obtained as $\bm{y} = \mathcal{A}(\bm{X}) + \boldsymbol{\epsilon}$, where $\bm{X}$ is the pixel matrix for the image of interest, and $\mathcal{A}$ is the collection of measurement masks for observing this image. Imaging systems of this form arise in many real-world applications, including single-pixel cameras \cite{Dea2008}, compressive hyperspectral imaging \cite{Rea2013}, and X-ray imaging \cite{BD2009}. Particularly for the latter two applications, the measurements $\bm{y}$ can be expensive to generate. Here, the proposed method can be used to \textit{design} the measurement masks, so that one can maximize image recovery performance given a certain budget constraint.
\item \textit{Text document indexing}: A key challenge in data mining is the sheer size of the database at hand (e.g., the large number of documents in text analysis), which greatly restricts the use of standard data analytic techniques. For large text databases, one solution is to compress (or \textit{index}) this database into a smaller, representative summary. This compression is typically performed by randomly projecting the large database onto a lower-dimensional subspace (see \cite{Ach2001, BM2001}). In other words, the summary data $\bm{y}$ is generated as $\bm{y} = \mathcal{A}(\bm{X})$, where $\mathcal{A}$ is a random linear projection operator. Here, the proposed method can be used to \textit{design} the projection operator $\mathcal{A}$, to achieve a good compression of the database $\bm{X}$.
\ei

In the past decade, there has been a rapidly growing body of literature on the topic of low-rank \textit{matrix recovery}, focusing largely on the theoretical properties of such a recovery via convex programming. This includes the seminal works of \cite{Rea2010} and \cite{CP2011a}, who investigated the necessary conditions for a successful recovery of $\bm{X}$ using nuclear-norm minimization, as well as numerous subsequent works (e.g., \cite{Oea2011, CZ2013}) which improved upon such conditions. A related problem, called \textit{matrix completion} (where matrix entries are directly observed), has received special attention; this includes the pioneering papers \cite{CR2009, CP2010, CT2010, Rec2011}, as well as many subsequent works. However, nearly all of the literature on matrix recovery (or matrix completion) assumes the measurement masks (or the missing entries) are sampled uniformly-at-random, since this allows for easy implementation and more amenable theoretical analysis. For the specific case of matrix completion, there has been some recent work on an informed (or \textit{designed}) strategy for sampling matrix entries \cite{Sea2013, Cea2013, MX2017}. To the best of our knowledge, no one has tackled the design problem for the more general setting of matrix recovery from a maximum entropy design perspective; this is the aim of the current paper.

We propose here a novel information-theoretic framework for designing the measurement masks in the linear operator $\mathcal{A}$, with the desired goal being an improved recovery of the low-rank matrix $\bm{X}$ compared to randomly sampled masks. The contributions of our work are three-fold. First, we derive a design principle called the \textit{maximum entropy principle} for the matrix recovery problem, which states that \textit{the measurement operator $\mathcal{A}$ maximizing the entropy of observations $\bm{y}$ is the operator maximizing information gain on matrix $\bm{X}$}. Such a principle yields a simple design criterion involving only the \textit{observations} $\bm{y}$, which serves as a proxy for more complicated criteria involving the \textit{matrix} $\bm{X}$ (the matrix $\bm{X}$ is typically much more high-dimensional than the observations $\bm{y}$, see, e.g., \cite{CP2011a}). Next, adopting a so-called singular matrix-variate Gaussian stochastic model on $\bm{X}$, we reveal novel insights between maximum entropy sampling and {subspace packings}, {by generalizing a lower bound on entropy under uniform subspace priors}. Using such insights, we then develop a novel algorithm, called \texttt{MaxEnt}, for efficiently designing initial and sequential measurement masks in $\mathcal{A}$. Finally, we demonstrate the effectiveness of the proposed design strategy in several numerical simulations, and in real-world applications on image processing and text document indexing.

There is a large body of work on information-theoretic design (e.g., for compressive sensing) and its many real-world applications (see, e.g., \cite{Mac1992}), and we would be remiss if we did not mention important developments in this field. This includes the seminal paper \cite{PV2006} {(see also \cite{Gea2005})}, who showed the profound fact that, for linear vector Gaussian channels, the gradient of the mutual information is related to the minimum mean-squared error matrix for parameter estimation. Such a result is further developed by \cite{Cea2012}, \cite{Wea2014} and \cite{Sea2017} for designing measurement matrices in compressive sensing and phase retrieval. The key {novelty in our work} is that, instead of directly maximizing the mutual information between signal (i.e., $\bm{X}$) and measurements (i.e., $\bm{y}$), we examine a dual (but equivalent) problem of maximizing the entropy of observations $\bm{y}$. As we show in the paper, this dual view sometimes offers the advantage of efficient initial and adaptive constructions of measurement masks via subspace packing, which then allows for effective matrix recovery. A related maximum entropy approach was also employed in \cite{FT2016} for developing a general minimax approach to supervised learning. Our work is also a novel extension of the information-theoretic matrix completion work \cite{MX2017}, in that we explore \textit{general} measurement masks (rather than \textit{entrywise} measurements).

This paper is organized as follows. Section 2 outlines the matrix recovery framework and the proposed model specification. Section 3 introduces the maximum entropy principle, and demonstrates how such a principle can be applied for designing measurement masks in $\mathcal{A}$. Sections 4 and 5 reveal new insights on initial and adaptive design, including a useful connection between maximum entropy masks and subspace packings. Section 6 details a design algorithm called \texttt{MaxEnt}, which can efficiently construct initial and adaptive masks for maximizing information gain on $\bm{X}$. Section 7 demonstrates the effectiveness of \texttt{MaxEnt} in numerical simulations, and explores its usefulness for solving real-world problems on image recovery and text document indexing. Section 8 concludes with thoughts on future work.

\section{Model framework for matrix recovery}
\label{sec:mod}
We begin by first outlining the matrix recovery problem, then reviewing the singular matrix-variate Gaussian model \cite{MX2017}. This model will serve as a versatile probabilistic model for low-rank matrices throughout the paper.

\subsection{Problem set-up}
Let $\bm{X} = (X_{i,j}) \in \mathbb{R}^{m_1 \times m_2}$ be the low-rank matrix of interest. Suppose $\bm{X}$ is observed via masks $\{\bm{A}_i\}_{i=1}^n \subseteq \mathbb{R}^{m_1 \times m_2}$, with the resulting samples then corrupted by independent and identically distributed (i.i.d.) Gaussian noise. The resulting noisy observations, $\{y_i\}_{i=1}^n$, then follow the model:
\begin{equation}
y_i = \mu_i + \epsilon_i, \; \mu_i = \langle \bm{A}_i, \bm{X} \rangle_F, \; \epsilon_i \distas{i.i.d.} \mathcal{N}(0,\eta^2), \; i = 1, \cdots, n,
\label{eq:mr}
\end{equation}
where $\langle \bm{A}, \bm{X} \rangle_F = \text{tr}(\bm{A}^T\bm{X})$ is the Frobenius inner-product. We assume all masks satisfy the unit power constraint $\|\bm{A}_i\|_F^2 \leq 1$, as is typical in matrix sensing problems \cite{DR2016}. With $\bm{y} = (y_i)_{i=1}^n$ and $\boldsymbol{\epsilon} = (\epsilon_i)_{i=1}^n$, \eqref{eq:mr} can be written in vector form:
\begin{equation}
\bm{y} = \mathcal{A}(\bm{X}) + \boldsymbol{\epsilon},
\label{eq:mrvec}
\end{equation}
where $\mathcal{A}: \mathbb{R}^{m_1 \times m_2} \rightarrow \mathbb{R}^n$ is the linear measurement operator returning the mean vector $\mathcal{A}(\bm{X}) = (\mu_i)_{i=1}^n$.

With this, the desired goal of mask design for matrix recovery can be made more precise. We employ here the following two-step design approach, commonly used in (statistical) experimental design \cite{WH2009}. First, given no prior knowledge on $\bm{X}$, \textit{initial} masks $\bm{A}_{1:n} = [\bm{A}_1 \; \bm{A}_2 \; \cdots \; \bm{A}_n]$ are designed to extract a maximum amount of initial information on $\bm{X}$. Next, from this initial learning, \textit{sequential} masks $\bm{A}_{n+1}, \bm{A}_{n+2}, \cdots$ are designed to adaptively maximize information on $\bm{X}$. The singular matrix-variate Gaussian distribution, presented below, provides an appealing framework for developing this two-step information-theoretic design methodology.

\subsection{Model specification}

\subsubsection{The singular matrix-variate Gaussian distribution}
Suppose $\bm{X}$ is normalized with zero mean, and consider the following model for $\bm{X}$:
\begin{definition}
[Singular matrix-variate Gaussian (SMG); Definition 2.4.1, \cite{GN1999}] Let $\bm{Z} \in \mathbb{R}^{m_1 \times m_2}$ be a random matrix with entries $Z_{i,j} \distas{i.i.d.} \mathcal{N}(0, \sigma^2)$ for $i = 1, \cdots, m_1$ and $j = 1, \cdots, m_2$. The random matrix $\bm{X}$ has a \textup{singular matrix-variate Gaussian} distribution if $\bm{X} \stackrel{d}{=} \mathcal{P}_{\mathcal{U}} \bm{Z} \mathcal{P}_{\mathcal{V}}$ for projection matrices $\mathcal{P}_{\mathcal{U}} = \bm{U}\bm{U}^T$ and $\mathcal{P}_{\mathcal{V}} = \bm{V}\bm{V}^T$, where $\bm{U} \in \mathbb{R}^{m_1 \times R}$, $\bm{U}^T\bm{U} = \bm{I}$,  $\bm{V} \in \mathbb{R}^{m_2 \times R}$, $\bm{V}^T\bm{V} = \bm{I}$ and $R < m_1 \wedge m_2$.\footnote{Here, $m_1 \wedge m_2 := \min(m_1,m_2)$ and $m_1 \vee m_2 := \max(m_1, m_2)$.} We will denote this as $\bm{X} \sim \mathcal{SMG}(\mathcal{P}_{\mathcal{U}},\mathcal{P}_{\mathcal{V}},\sigma^2,R)$.
\label{def:smg}
\end{definition}
\noindent From a simulation perspective, a realization from $\mathcal{SMG}(\mathcal{P}_{\mathcal{U}},\mathcal{P}_{\mathcal{V}},\sigma^2,R)$ is obtained by first (a) simulating a random matrix $\bm{Z}$ with each entry following an i.i.d. $\mathcal{N}(0,\sigma^2$) distribution, then (b) performing a left and right projection of $\bm{Z}$ via the projection matrices $\mathcal{P}_{\mathcal{U}}$ and $\mathcal{P}_{\mathcal{V}}$. Here, $\mathcal{P}_{\mathcal{U}}$ and $\mathcal{P}_{\mathcal{V}}$ are projection operators which map vectors from $\mathbb{R}^{m_1}$ and $\mathbb{R}^{m_2}$ onto $\mathcal{U}$ and $\mathcal{V}$, the $R$-dim. linear subspaces spanned by the orthonormal columns of $\bm{U}$ and $\bm{V}$, respectively. After this left-right projection of $\bm{Z}$, one can show that the resulting matrix $\bm{X} = \mathcal{P}_{\mathcal{U}}\bm{Z}\mathcal{P}_{\mathcal{V}}$ has rank $R < m_1 \wedge m_2$, with its row and column spaces lying in $\mathcal{U}$ and $\mathcal{V}$, respectively. For $R$ small, the SMG model provides a flexible framework for modeling low-rank matrices.

Similar to the design problem in matrix completion \cite{MX2017}, the parametrization of the SMG model using projection matrices offers several appealing features for mask design in matrix recovery. First, such a parametrization encodes valuable information on the subspaces of $\bm{X}$. Recall that each projection operator $\mathcal{P}_{\mathcal{W}} \in \mathbb{R}^{m \times m}$ of rank $R$ corresponds to a unique $R$-plane $\mathcal{W}$ (i.e., an $R$-dim. linear subspace) in $\mathbb{R}^m$. $\mathcal{P}_{\mathcal{U}}$ and $\mathcal{P}_{\mathcal{V}}$ then parametrize information on the row space $\mathcal{U} \in \mathcal{G}_{R,m_1-R}$ and the column space $\mathcal{V} \in \mathcal{G}_{R,m_2-R}$, where $\mathcal{G}_{R,m-R}$ is the \textit{Grassmann manifold} consisting of all $R$-planes in $\mathbb{R}^m$. This can then be used to derive insightful connections between initial mask design and Grassmann packings (see Section \ref{sec:ini}). Second, using the fact that the projection of a Gaussian random vector is still Gaussian-distributed, we show later that, conditional on observations $\bm{y}$, subsequent observations will also be Gaussian-distributed. This property is key for deriving a closed-form adaptive design scheme for greedily maximizing information on $\bm{X}$ (see Section \ref{sec:seq}).

\subsection{Connection to nuclear-norm recovery}
\label{sec:connuc}
For most practical scenarios, there is little-to-no prior knowledge on either the rank of $\bm{X}$ or its subspaces. In such cases, a Bayesian approach \cite{Gea2014} would be to assign non-informative prior distributions to the model parameters $R$, $\mathcal{P}_{\mathcal{U}}$ and $\mathcal{P}_{\mathcal{V}}$, i.e., by assuming all possible ranks and subspaces are equally likely. Adopting these non-informative priors, the following lemma reveals an insightful expression for the maximum-a-posteriori (MAP) estimator of $\bm{X}$:
\begin{lemma}[MAP estimator]
Assume $\eta^2$ and $\sigma^2$ are fixed. Suppose (a) the subspaces for $\mathcal{P}_{\mathcal{U}}$ and $\mathcal{P}_{\mathcal{V}}$ are uniformly distributed over the Grassmann manifolds $\mathcal{G}_{R,m_1-R}$ and $\mathcal{G}_{R,m_2-R}$, and (b) $R$ is uniformly distributed on $\{1, \cdots, m_1 \wedge m_2\}$. Conditional on $\bm{y}$, the \textup{MAP estimator} for $\bm{X}$ becomes:
\begin{equation}
\small
\widetilde{\bm{X}} \in \underset{\bm{X} \in \mathbb{R}^{m_1 \times m_2}}{\textup{argmin}} \; \left[ \frac{\| \bm{y} - \mathcal{A}(\bm{X})\|_2^2}{\eta^2} + \log(2 \pi \sigma^2) \textup{rank}^2(\bm{X}) + \frac{\|\bm{X}\|_F^2}{\sigma^2} \right],
\label{eq:matest}
\end{equation}
\normalsize
where $\|\bm{X}\|_F = \sqrt{\sum_{i,j} X_{i,j}^2}$ is the Frobenius norm of $\bm{X}$.
\label{thm:mle}
\end{lemma}

Consider now the following approximation of \eqref{eq:matest}. First, viewing the rank penalty $\log(2 \pi \sigma^2) \textup{rank}^2(\bm{X})$ as a Lagrange multiplier, this penalty term can be replaced by the constraint $\textup{rank}(\bm{X}) \leq \sqrt{\xi}$. Next, changing this constraint into its Lagrangian form, and relaxing the rank function $\textup{rank}(\bm{X})$ into the nuclear norm $\|\bm{X}\|_*$, which is the tightest convex relaxation \cite{Rea2010}, \eqref{eq:matest} becomes:
\begin{equation}
\widetilde{\bm{X}} = \underset{\bm{X} \in \mathbb{R}^{m_1 \times m_2}}{\textup{argmin}} \; \Big[ \| \bm{y} - \mathcal{A}(\bm{X})\|_2^2 + \lambda \left\{ \alpha \|\bm{X}\|_* + (1-\alpha) \|\bm{X}\|_F^2 \right\} \Big],
\label{eq:el}
\end{equation}
for an appropriate choice of $\lambda > 0$ and $\alpha \in (0,1)$. The problem in \eqref{eq:el} can then be viewed as an ``elastic net'' \cite{ZH2005} formulation for low-rank matrix recovery.

The formulation in \eqref{eq:el} yields an interesting connection between the MAP estimator $\widetilde{\bm{X}}$ and existing recovery methods. Setting $\alpha = 1$, \eqref{eq:el} reduces to the nuclear-norm formulation:
\begin{equation}
\widehat{\bm{X}} = \underset{\bm{X} \in \mathbb{R}^{m_1 \times m_2}}{\textup{argmin}} \; \Big[ \| \bm{y} - \mathcal{A}(\bm{X})\|_2^2 + \lambda \|\bm{X}\|_* \Big],
\label{eq:nuc}
\end{equation}
which is widely used for low-rank matrix recovery \cite{CP2011,PB2014}. This link allows us to use efficient algorithms for solving \eqref{eq:nuc} to guide the active mask design procedure (see Section \ref{sec:alg}).

%

\subsection{Useful model properties}
The SMG model also offers several properties which will prove useful later for mask design. These properties concern the joint distribution of measurements, both prior to and conditional on obtaining observations from $\bm{X}$.

\subsubsection{Joint distribution prior to observations}
\label{sec:priorobs}
Let $y_i$ and $y_j$ be observations from \eqref{eq:mr} using masks $\bm{A}_i$ and $\bm{A}_j$, respectively. The following lemma gives a closed-form joint distribution for $y_i$ and $y_j$, prior to observing data on $\bm{X}$:
\begin{lemma}[Unconditional joint distribution]
Suppose $\bm{X} \sim \mathcal{SMG}(\mathcal{P}_{\mathcal{U}}, \mathcal{P}_{\mathcal{V}},\sigma^2,R)$, with $\mathcal{P}_{\mathcal{U}}$, $\mathcal{P}_{\mathcal{V}}$, $\sigma^2$ and $R$ fixed. Then:
\begin{equation}
\begin{bmatrix}
y_i\\
y_j
\end{bmatrix} \sim \mathcal{N}\left\{ 
\begin{bmatrix}
0 \\ 0
\end{bmatrix}, 
\begin{pmatrix}
\xi_{i,i} & \xi_{i,j}\\
\xi_{i,j} & \xi_{j,j}
\end{pmatrix}
\right\}, \quad i \neq j,
\label{eq:uncdist}
\end{equation}
where:
\begin{align}
\begin{split}
\xi_{i,i} &= \textup{Var}(y_i) = \sigma^2 \| \mathcal{P}_{\mathcal{U}} \bm{A}_i \mathcal{P}_{\mathcal{V}} \|_F^2 + \eta^2,\\
\xi_{i,j} &= \textup{Cov}(y_i, y_j) = \sigma^2 \langle \mathcal{P}_{\mathcal{U}} \bm{A}_i \mathcal{P}_{\mathcal{V}}, \mathcal{P}_{\mathcal{U}} \bm{A}_j \mathcal{P}_{\mathcal{V}} \rangle_F.
\label{eq:var}
\end{split}
\end{align}
\label{lem:var}
\end{lemma}

These closed-form variance and covariance expressions can be viewed as similarity measures between measurement masks and the subspaces of $\bm{X}$. Consider first the variance expression for $\text{Var}(y_i)$, which (ignoring measurement noise $\eta^2$) is proportional to $\|\mathcal{P}_{\mathcal{U}} \bm{A}_i \mathcal{P}_{\mathcal{V}}\|_F^2$. Viewed geometrically, the variance of an observation from mask $\bm{A}_i$ is proportional to the norm of $\bm{A}_i$, \textit{after} accounting for its \textit{similarity} with the subspaces of $\bm{X}$ via a left-right projection by $\mathcal{P}_{\mathcal{U}}$ and $\mathcal{P}_{\mathcal{V}}$. Consider next the covariance between $y_i$ and $y_j$, which is proportional to the inner-product $\langle \mathcal{P}_{\mathcal{U}} \bm{A}_i \mathcal{P}_{\mathcal{V}}, \mathcal{P}_{\mathcal{U}} \bm{A}_j \mathcal{P}_{\mathcal{V}} \rangle_F$. This can be seen as the \textit{angle} between the two masks $\bm{A}_i$ and $\bm{A}_j$, \textit{after} accounting for their similarities with the subspaces of $\bm{X}$.

\begin{figure}
\centering
\includegraphics[width=0.3\textwidth]{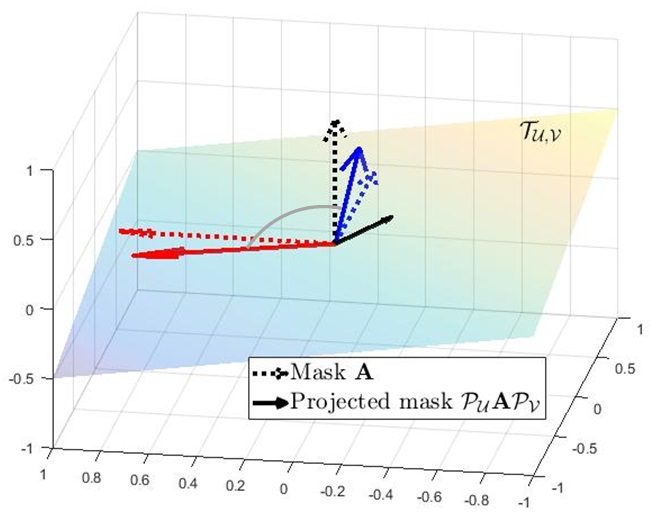}
\caption{A visualization of three masks (dotted vectors) on the Frobenius inner-product space $\langle \cdot, \cdot \rangle_F$, and its projections (solid vectors) onto subspace $\mathcal{T}_{\mathcal{U},\mathcal{V}}$ (corresponding to the row and column spaces of $\bm{X}$).}
\label{fig:proj}
\vspace{-0.2in}
\end{figure}

Figure \ref{fig:proj} visualizes this geometric interpretation. Here, the shaded subspace $\mathcal{T}_{\mathcal{U},\mathcal{V}}$ represents the projected subspace from a left-right projection by $\mathcal{P}_{\mathcal{U}}$ and $\mathcal{P}_{\mathcal{V}}$ (further details on $\mathcal{T}_{\mathcal{U},\mathcal{V}}$ in Lemma \ref{lem:oproj}), the three dotted vectors represent three measurement masks in the Frobenius inner-product space $\langle \cdot, \cdot \rangle_F$, and the three solid vectors represent the projection of these masks onto $\mathcal{T}_{\mathcal{U},\mathcal{V}}$. The variance term $\textup{Var}(y_i) \propto \|\mathcal{P}_{\mathcal{U}} \bm{A}_i \mathcal{P}_{\mathcal{V}}\|_F^2$ (again, ignoring $\eta^2$) is the squared-\textit{length} of the projected vector (solid) for mask $\bm{A}_i$. Comparing the three projected vectors in Figure \ref{fig:proj}, we see that the mask with greater similarity to $\mathcal{U}$ and $\mathcal{V}$ (the red vector) has a longer projected vector, so observations from this mask have greater variance. Likewise, the covariance term $ \text{Cov}(y_i,y_j) \propto \langle \mathcal{P}_{\mathcal{U}} \bm{A}_i \mathcal{P}_{\mathcal{V}}, \mathcal{P}_{\mathcal{U}} \bm{A}_j \mathcal{P}_{\mathcal{V}} \rangle_F$ is the \textit{inner-product} between the projected vectors for two masks $\bm{A}_i$ and $\bm{A}_j$. A larger inner-product (or smaller angle) between these two projected vectors suggests higher correlations between observations from masks $\bm{A}_i$ and $\bm{A}_j$. As shown later in Section \ref{sec:ini}, the goal of designing masks which \textit{maximize} information on $\bm{X}$ can be viewed as finding masks which \textit{maximize} the angles between their projected vectors.


\subsubsection{Joint distribution conditional on observations}
\label{sec:condobs}
Now, suppose the observations $\bm{y}$ have been sampled from \eqref{eq:mr}, and let $y_{n+1}$ and $y_{n+2}$ be new observations from masks $\bm{A}_{n+1}$ and $\bm{A}_{n+2}$. Using the conditional property of the Gaussian distribution, the following lemma provides a closed-form joint distribution for $y_{n+1}$ and $y_{n+2}$, conditional on $\bm{y}$:

\begin{lemma}[Conditional joint distribution]
Let $\bm{y}$ be observations from masks $\bm{A}_{1:n} = [\bm{A}_1 \; \bm{A}_2 \; \cdots \; \bm{A}_n]$, and let $y_{n+1}$ and $y_{n+2}$ be new observations from masks $\bm{A}_{n+1}$ and $\bm{A}_{n+2}$. Assuming $\bm{X} \sim \mathcal{SMG}(\mathcal{P}_{\mathcal{U}}, \mathcal{P}_{\mathcal{V}},\sigma^2,R)$, with $\mathcal{P}_{\mathcal{U}}$, $\mathcal{P}_{\mathcal{V}}$, $\sigma^2$ and $R$ fixed, we have:
\begin{equation}
\small
\begin{bmatrix}
y_{n+1}\\
y_{n+2}
\end{bmatrix}\Big| \bm{y} \sim \mathcal{N}\left\{ 
\begin{bmatrix}
\mathbb{E}(y_{n+1}|\bm{y})\\
\mathbb{E}(y_{n+2}|\bm{y})
\end{bmatrix}, 
\begin{pmatrix}
\xi_{n+1,n+1}|\bm{y} & \xi_{n+1,n+2}|\bm{y}\\
\xi_{n+1,n+2}|\bm{y} & \xi_{n+2,n+2}|\bm{y}
\end{pmatrix}
\right\},
\label{eq:conddist}
\normalsize
\end{equation}
where, with:
\begin{align}
\begin{split}
\bm{R}_n(\bm{A}_{1:n}) &:= [\langle \mathcal{P}_{\mathcal{U}} \bm{A}_i \mathcal{P}_{\mathcal{V}}, \mathcal{P}_{\mathcal{U}} \bm{A}_j \mathcal{P}_{\mathcal{V}} \rangle_F]_{i,j=1}^n \in \mathbb{R}^{n \times n},\\
\bm{r}_n(\bm{A}) &:= [\langle \mathcal{P}_{\mathcal{U}} \bm{A}_i \mathcal{P}_{\mathcal{V}}, \mathcal{P}_{\mathcal{U}} \bm{A} \mathcal{P}_{\mathcal{V}} \rangle_F]_{i=1}^n \in \mathbb{R}^{n},
\label{eq:rn}
\end{split}
\end{align}
and $\gamma^2 := \eta^2 / \sigma^2$, we have:
\begin{align}
\small
\begin{split}
\mathbb{E}(y_{i}|\bm{y}) &=  \bm{r}_n^T(\bm{A}_{i}) [\bm{R}_n(\bm{A}_{1:n}) + \gamma^2 \bm{I}]^{-1} \bm{y},\\
\xi_{i,i}|\bm{y} &:= \textup{Var}(y_{i}|\bm{y})\\
&= \textup{Var}(y_{i}) - \sigma^2 \bm{r}_n^T(\bm{A}_{i}) [\bm{R}_n(\bm{A}_{1:n}) + \gamma^2 \bm{I}]^{-1} \bm{r}_n(\bm{A}_{i}),\\
\xi_{i,j}|\bm{y} &:= \textup{Cov}(y_{i},y_{j}|\bm{y})\\
& = \textup{Cov}(y_{i},y_{j}) - \sigma^2 \bm{r}_n^T(\bm{A}_{i}) [\bm{R}_n(\bm{A}_{1:n}) + \gamma^2 \bm{I}]^{-1} \bm{r}_n(\bm{A}_{j}).
\label{eq:condpar}
\end{split}
\normalsize
\end{align}
\label{thm:conddist}
\end{lemma}

These closed-form conditional expressions also enjoy intuitive interpretations. In particular, the \textit{conditional} variance of a new observation, $\textup{Var}(y_{n+1}|\bm{y})$, can be decomposed as the \textit{unconditional} variance $\textup{Var}(y_{n+1})$, minus a \textit{reduction} term quantifying how correlated the new mask $\bm{A}_{n+1}$ is to the observed masks $\bm{A}_{1:n}$. Similarly, the \textit{conditional} covariance between two new observations, $\textup{Cov}(y_{n+1},y_{n+2}|\bm{y})$, can be decomposed as the \textit{unconditional} covariance $\textup{Cov}(y_{n+1},y_{n+2})$, minus a \textit{reduction} term quantifying how correlated the new masks $\bm{A}_{n+1}$ and $\bm{A}_{n+2}$ are to the observed masks $\bm{A}_{1:n}$. The expressions in \eqref{eq:condpar} will reappear when deriving the sequential mask design procedure in Section \ref{sec:seq}.

\section{An information-theoretic view on mask design}
\label{sec:info}
Next, we present an information-theoretic mask design framework for matrix recovery. We first outline the principle of maximum entropy, then show how such a principle can be used to design masks which maximize information gain on $\bm{X}$.

\subsection{The design principle of maximum entropy}
The principle of maximum entropy was first introduced in \cite{SW1987} and further developed in \cite{SW2000} for (statistical) experimental design in spatio-temporal modeling, although the origins of information-theoretic experimental design date back much further to the seminal works of \cite{Bla1951} and \cite{Lin1956}. This principle states that, under regularity assumptions on an observation model with unknown model parameters, \textit{a design scheme maximizing the entropy of collected data is a design scheme maximizing information gain on model parameters}. In other words, to maximize information on unknown parameters, one should sample from a design scheme which maximizes the entropy of collected samples. As described in \cite{SW2000}, the maximum entropy principle offers two important advantages for design. First, it allows for \textit{efficient} design construction, since the entropy expression for observed data is oftentimes simple and closed-form. Second, the trade-off between sample entropy and model information reveals useful insights. We will demonstrate how such advantages play a role in matrix recovery mask design.


To formally present this principle, we require the following definitions. Let $(X,Y)$ be a pair of random variables (r.v.s) with marginal densities $(f_X(x), f_Y(y))$ and joint density $f_{X,Y}(x,y)$. Following \cite{CT2012}, the \textit{entropy} of $X$ is defined as ${\rm H}(X) = \mathbb{E}[ - \log f_X(X)]$, with larger values indicating greater uncertainty for $X$. Similarly, the \textit{joint entropy} of $(X,Y)$ is ${\rm H}(X,Y) = \mathbb{E}[ - \log f_{X,Y}(X,Y)]$, and the \textit{conditional entropy} of $Y$ given $X$ is the entropy of the conditional r.v. $Y|X$, which we denote by ${\rm H}(Y|X)$. The following chain rule (Theorem 2.2.1 in \cite{CT2012}) provides a link between joint entropy ${\rm H}(X,Y)$ and conditional entropy ${\rm H}(Y|X)$: 
\begin{equation}
{\rm H}(X,Y) = {\rm H}(X) + {\rm H}(Y|X).
\label{eq:ch}
\end{equation}

With this in hand, consider now the matrix recovery problem. Here, the parameter-of-interest is the unknown low-rank matrix $\bm{X}$, the design scheme is the choice of measurement masks $\bm{A}_{1:n} = [\bm{A}_1 \; \bm{A}_2 \; \cdots \; \bm{A}_n]$, and the collected data are the observations $\bm{y}_{\bm{A}_{1:n}}$ taken from these masks. Using the chain rule in \eqref{eq:ch}, we get the following decomposition:
\begin{equation}
{\rm H}(\bm{y}_{\bm{A}_{1:n}},\bm{X}) = {\rm H}(\bm{y}_{\bm{A}_{1:n}}) + {\rm H}(\bm{X}|\bm{y}_{\bm{A}_{1:n}}),
\label{eq:chmat}
\end{equation}
The first term in \eqref{eq:chmat} is the joint entropy of observations $\bm{y}_{\bm{A}_{1:n}}$ and matrix $\bm{X}$, the middle term ${\rm H}(\bm{y}_{\bm{A}_{1:n}})$ is the entropy of observations $\bm{y}_{\bm{A}_{1:n}}$ from masks $\bm{A}_{1:n}$, and the last term ${\rm H}(\bm{X}|\bm{y}_{\bm{A}_{1:n}})$ is the conditional entropy of $\bm{X}$ after observing $\bm{y}_{\bm{A}_{1:n}}$. Our goal is to design masks $\bm{A}_{1:n}$ which \textit{minimize} the conditional entropy ${\rm H}(\bm{X}|\bm{y}_{\bm{A}_{1:n}})$, thereby \textit{maximizing} the information gained on $\bm{X}$ after observing $\bm{y}_{\bm{A}_{1:n}}$.

The maximum entropy principle can then be derived as follows. Applying the chain rule again to the joint entropy ${\rm H}(\bm{y}_{\bm{A}_{1:n}},\bm{X})$ in \eqref{eq:chmat}, we get:
\begin{align}
{\rm H}(\bm{y}_{\bm{A}_{1:n}},\bm{X}) &= {\rm H}(\bm{X}) + {\rm H}(\bm{y}_{\bm{A}_{1:n}}|\bm{X}) \tag{by \eqref{eq:ch}}\\
& = {\rm H}(\bm{X}) + {\rm H}(\mathcal{A}(\bm{X}) + \boldsymbol{\epsilon}|\bm{X}) \tag{by \eqref{eq:mrvec}}\\
& = {\rm H}(\bm{X}) + {\rm H}(\boldsymbol{\epsilon}|\bm{X}) \tag{$\mathcal{A}(\bm{X})$ is fixed given $\bm{X}$}\\
& = {\rm H}(\bm{X}) + {\rm H}(\boldsymbol{\epsilon}). \tag{$\boldsymbol{\epsilon}$ and $\bm{X}$ are independent}
\end{align}
The key observation here is that the final quantity ${\rm H}(\bm{X}) + {\rm H}(\boldsymbol{\epsilon})$ -- the sum of entropies for matrix $\bm{X}$ and measurement noise $\boldsymbol{\epsilon}$ -- \textit{does not depend on the choice of masks $\bm{A}_{1:n}$}. Returning to \eqref{eq:chmat}, this means the left-hand side of \eqref{eq:chmat} also does not depend on $\bm{A}_{1:n}$, and hence the masks $\bm{A}_{1:n}$ which \textit{minimize} ${\rm H}(\bm{X}|\bm{y}_{\bm{A}_{1:n}})$ in \eqref{eq:chmat} also \textit{maximize} ${\rm H}(\bm{y}_{\bm{A}_{1:n}})$ as well. This is precisely the maximum entropy principle for matrix recovery -- \textit{a mask design which maximizes the entropy of observations $\bm{y}_{\bm{A}_{1:n}}$ in turn maximizes information gain on $\bm{X}$}. Computationally, such a principle allows us to employ the simpler entropy term ${\rm H}(\bm{y}_{\bm{A}_{1:n}})$ as an efficient proxy for the desired entropy term ${\rm H}(\bm{X}|\bm{y}_{\bm{A}_{1:n}})$, which is much more complicated and difficult to minimize in high-dimensions.

The maximization of observation entropy ${\rm H}(\bm{y}_{\bm{A}_{1:n}})$ may also lead to reductions in recovery error for $\bm{X}$, which is ultimately the desired goal. By the maximum entropy principle, maximizing ${\rm H}(\bm{y}_{\bm{A}_{1:n}})$ is equivalent to minimizing the conditional entropy ${\rm H}(\bm{X}|\bm{y}_{\bm{A}_{1:n}})$. The following lower bound (Equation 27 in \cite{Pra2010}) then connects this conditional entropy with expected recovery error $\mathbb{E}[\|\bm{X}-\tilde{\bm{X}}\|_F^2|\bm{y}_{\bm{A}_{1:n}}]$, $\tilde{\bm{X}} = \mathbb{E}[\bm{X}|\bm{y}_{\bm{A}_{1:n}}]$:
\begin{equation}
\mathbb{E}[\|\bm{X}-\tilde{\bm{X}}\|_F^2|\bm{y}] \geq \frac{1}{2\pi e}\exp\left\{ 2{\rm H}(\bm{X}|\bm{y}_{\bm{A}_{1:n}}) \right\}.
\label{eq:entmmse}
\end{equation}
In other words, by designing masks which maximize observation entropy, one hopes to achieve lower expected recovery errors on $\bm{X}$ (see \cite{Cea2012,Dea2013,Wea2017} for further justification). As before, the key advantage in working with observational entropy ${\rm H}(\bm{y}_{\bm{A}_{1:n}})$ is that it offers a simple and insightful expression for mask construction, whereas the error term $\mathbb{E}[\|\bm{X}-\tilde{\bm{X}}\|_F^2|\bm{y}]$ is much more cumbersome to optimize, particularly in high-dimensions.


\subsection{Maximum entropy masks}
\label{sec:maxent}
{We now explore further the properties of these ``maximum entropy masks'', i.e., the masks $\bm{A}_{1:n}$ which maximize observation entropy ${\rm H}(\bm{y}_{\bm{A}_{1:n}})$. Unfortunately, when the true subspaces ($\mathcal{U}, \mathcal{V}$) are unknown (and assumed to be random), this entropy term cannot be evaluated in closed form. To derive a closed-form expression, suppose for now fixed subspaces ($\mathcal{U}, \mathcal{V}$); this will be relaxed later. Let ${\rm E}_{\mathcal{U}, \mathcal{V}}(\bm{y}_{\bm{A}_{1:n}}) := \exp\{{\rm H}(\bm{y}_{\bm{A}_{1:n}}|\mathcal{P}_{\mathcal{U}},\mathcal{P}_{\mathcal{V}})\}$ be the exponential of the conditional entropy\footnote{{Here, the exponential entropy (or exp-entropy) allows for closed form derivations throughout the paper. The maximum entropy principle holds for either entropy or exp-entropy, since the exponential function is monotone.}} for observations $\bm{y}_{\bm{A}_{1:n}}$, given subspaces ($\mathcal{U}, \mathcal{V}$). Applying Lemma \ref{lem:var}, we obtain the following simple expression for ${\rm E}_{\mathcal{U}, \mathcal{V}}(\bm{A}_{1:n})$:
\begin{equation}
{\rm E}_{\mathcal{U}, \mathcal{V}}(\bm{A}_{1:n}) \propto \det\{\sigma^2 \bm{R}_n(\bm{A}_{1:n}) + \eta^2 \bm{I}\},
\label{eq:ent}
\end{equation}
where $\bm{R}_n(\bm{A}_{1:n}) = [\langle \mathcal{P}_{\mathcal{U}} \bm{A}_i \mathcal{P}_{\mathcal{V}}, \mathcal{P}_{\mathcal{U}} \bm{A}_j \mathcal{P}_{\mathcal{V}} \rangle_F]_{i,j=1}^n$ is the matrix of inner-products in \eqref{eq:rn}. This closed form is a direct result of the Gaussian property of $\bm{X}$ from the SMG model.}

\begin{figure}[t]
\centering
\includegraphics[width=0.3\textwidth]{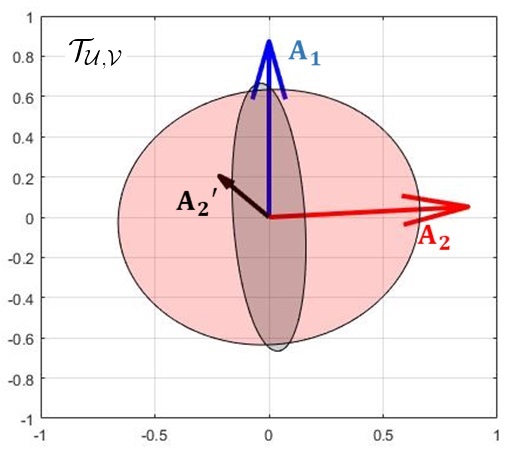}
\caption{The covariance matrices for the red and blue masks (red ellipse), and for the black and blue masks (black ellipse) from Figure \ref{fig:proj}.}
\label{fig:det}
\vspace{-0.2in}
\end{figure}

{The masks which maximize exp-entropy ${\rm E}_{\mathcal{U}, \mathcal{V}}(\bm{A}_{1:n})$ in \eqref{eq:ent} enjoy a geometric interpretation. Here, ${\rm E}_{\mathcal{U}, \mathcal{V}}(\bm{A}_{1:n})$ can be viewed as the \textit{volume} of the covariance matrix formed by the projected masks $\{\mathcal{P}_{\mathcal{U}} \bm{A}_i \mathcal{P}_{\mathcal{V}}\}_{i=1}^n$ on $\mathcal{T}_{\mathcal{U},\mathcal{V}}$. Figure \ref{fig:det} visualizes this using the previous example in Figure \ref{fig:proj}. Here, the red ellipse corresponds to the covariance matrix for the red and blue projected vectors, and the black ellipse corresponds to the covariance matrix for the black and blue projected vectors. Clearly, the red ellipse has a larger volume than the black ellipse; by sampling the red and blue masks, the resulting observations then have greater entropy than that from the black and blue masks. By the maximum entropy principle, the former yields more information on $\bm{X}$ than the latter. Viewed this way, the maximum entropy masks can be seen as masks which maximize the volume of their covariance matrix, after accounting for its similarity with the subspaces of $\bm{X}$.

We would like to comment that, for the \textit{initial} mask design of $\bm{A}_{1:n}$, the closed-form exp-entropy in \eqref{eq:ent} is \textit{not} a suitable optimization criterion, because one typically has little-to-no information on subspaces ($\mathcal{U}, \mathcal{V}$) prior to data. (In the rare instance where ($\mathcal{U}, \mathcal{V}$) are known with certainty prior to data, the masks maximizing \eqref{eq:ent} can be obtained from the first $n$ principal components of $\text{Var}\{\text{vec}(\bm{X})\}$; see Section \ref{sec:seq} for details.) Our strategy for tackling this problem is as follows: we will first derive a lower bound (Theorem \ref{thm:inient}) on the conditional exp-entropy ${\rm E}_{\mathcal{U}, \mathcal{V}}(\bm{A}_{1:n})$ via its closed form \eqref{eq:ent}, then make a connection to Grassmann packings by generalizing this bound under uniform priors on ($\mathcal{U}, \mathcal{V}$). This motivates an efficient initial mask construction via subspace packings, with the resulting masks approximately maximizing the (unconditional) observation entropy ${\rm H}(\bm{y}_{\bm{A}_{1:n}})$, which is complicated and has no closed form. Section \ref{sec:ini} provides further details on this construction.

For the \textit{sequential} mask design of $\bm{A}_{n+1}, \bm{A}_{n+2}, \cdots$ given initial masks $\bm{A}_{1:n}$, the exp-entropy in \eqref{eq:ent} again requires modification. First, as more data are collected on $\bm{X}$, more accurate estimates can be obtained on subspaces ($\mathcal{U}, \mathcal{V}$); a sequential scheme should incorporate this \textit{adaptive} learning on subspaces to guide active sampling. Second, a good sequential design strategy should encourage subsequent masks to sample directions which are \textit{unexplored} by prior masks $\bm{A}_{1:n}$. Our strategy is as follows: we will first use the nuclear-norm minimization in \eqref{eq:nuc} to obtain subspace estimates $(\hat{\mathcal{U}}_n,\hat{\mathcal{V}}_n)$, then employ a sequential modification of \eqref{eq:ent} with plug-in estimates $(\hat{\mathcal{U}}_n,\hat{\mathcal{V}}_n)$ to derive an efficient, adaptive design algorithm. From Lemma \ref{thm:mle}, this can be seen as an (approximate) MAP-guided active sampling algorithm for matrix recovery; more on this in Sections \ref{sec:seq} and \ref{sec:seqc}.}

%

\section{Insights on initial mask design}
\label{sec:ini}
Consider first the \textit{initial} design problem for masks $\bm{A}_{1:n}$, given no prior knowledge on the subspaces of $\bm{X}$. {To reflect this lack of knowledge, we make two intuitive assumptions:
\bi[leftmargin=10pt]
\item \textbf{(A1)}: Independent, uniform priors on $\mathcal{P}_{\mathcal{U}}$ and $\mathcal{P}_{\mathcal{V}}$ over the Grassmann manifolds $\mathcal{G}_{R,m_1 - R}$ and $\mathcal{G}_{R,m_2-R}$, respectively.
\item \textbf{(A2)}: Initial masks $\bm{A}_{1:n}$ follow the singular-value decomposition (SVD) form:
\begin{equation}
\bm{A}_i = \bm{R}_i \boldsymbol{\Lambda}_i \bm{S}_i^T, \quad i = 1, \cdots, n,
\label{eq:inistr}
\end{equation}
where the weight matrices follow $\boldsymbol{\Lambda}_i = R^{-1/2} \bm{I}$.
\ei
Assumption \textbf{(A1)} reflects the prior belief that all subspaces in $\bm{X}$ are equally likely before observing data. Assumption \textbf{(A2)} reflects the belief that all subspaces are weighed equally prior to data. Here, the scaling factor $R^{-1/2}$ on $\boldsymbol{\Lambda}_i$ ensures the unit power constraint is satisfied.}

Under \textbf{(A1)} and \textbf{(A2)}, we show the problem of designing initial masks under maximum entropy is related to the problem of subspace packings. We first provide a brief review of block coherence and subspace packings, then derive the link between initial design and subspace packings via a lower bound on \eqref{eq:ent}.

\subsection{Block coherence and subspace packings}
Define first an \textit{$R$-frame} in $\mathbb{R}^{m}$ (see \cite{BM2012}) -- a matrix $\bm{F} \in \mathbb{R}^{m \times R}$ with orthonormal columns, i.e., $\bm{F}^T\bm{F} = \bm{I}$. We employ two metrics to quantify the ``closeness'' between two frames, both of which are defined below:
\begin{definition}[Worst-case and avg. block coherence; \cite{BM2012}]
Let $\bm{F}_{1:n} = [\bm{F}_1 \; \bm{F}_2 \; \cdots \; \bm{F}_n] \in \mathbb{R}^{m \times nR}$ be a collection of $R$-frames in $\mathbb{R}^{m}$, where $m \geq 2R$, and let $\| \cdot \|_2$ be the spectral norm. The \textup{worst-case block coherence} of $\bm{F}_{1:n}$ is defined as:
\begin{equation}
\mu(\bm{F}_{1:n}) := \max_{i \neq j} \|\bm{F}_i^T \bm{F}_j\|_2,
\label{eq:wccoh}
\end{equation}
and the \textup{average block coherence} of $\bm{F}_{1:n}$ is defined as:
\begin{equation}
a(\bm{F}_{1:n}) := \frac{1}{n-1} \max_i \Big\|\sum_{j:j \neq i} \bm{F}_i^T \bm{F}_j\Big\|_2.
\label{eq:avcoh}
\end{equation}
\label{def:blcoh}
\end{definition}
\noindent Both the worst-case block coherence $\mu(\bm{F}_{1:n})$ and average block coherence $a(\bm{F}_{1:n})$ play a role in quantifying the recovery performance of block compressive sensing methods \cite{Eea2010}; the lower the coherence, the easier recovery becomes.

\begin{figure}[t]
\centering
\includegraphics[width=0.3\textwidth]{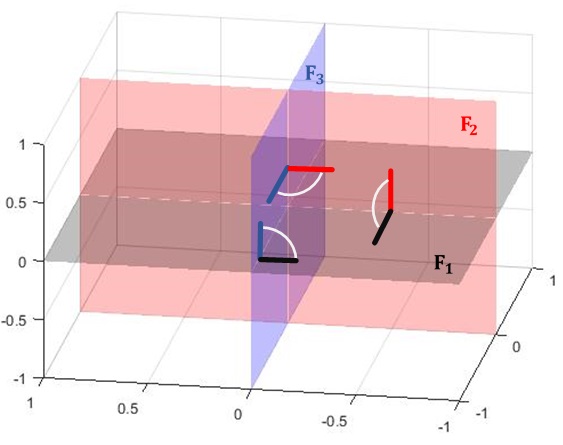}
\caption{A visualization of an optimal Grassmann packing for $n=3$ frames, each in $R=2$ dimensions. White arcs denote principal angles between any two of the three subspaces.}
\label{fig:pack}
\vspace{-0.2in}
\end{figure}

These coherence metrics also have an appealing geometric connection to the problem of optimal \textit{Grassmann packings} -- the packing of $R$-dim. subspaces in $\mathbb{R}^m$. In particular, \cite{Cea2015} shows that the subspace packing which \textit{maximizes} the minimum {nonzero} principal angle between any two subspaces, corresponds to frames $\bm{F}_{1:n}$ which \textit{minimize} the worst-case block coherence $\mu(\bm{F}_{1:n})$. Figure \ref{fig:pack} visualizes this connection using $n=3$ frames, each in $R=2$ dimensions. One sees that the principal angles between any two of the three subspaces are \textit{maximized}, so the frames $\bm{F}_{1:3}$ for these subspaces \textit{minimize} the worst-case coherence $\mu(\bm{F}_{1:3})$. In the same way, the frames which minimize average coherence $a(\bm{F}_{1:n})$ corresponds to a packing which minimizes some {averaged function} of the principal angles between subspaces.

\subsection{Initial masks and subspace packings}
\label{sec:inicoh}
With this in hand, we can now establish an interesting connection between maximum entropy masks and the block {coherence of their corresponding} frames, under the assumption of no prior knowledge on $\bm{X}$. Consider first a lower bound on the conditional exp-entropy ${\rm E}_{\mathcal{U},\mathcal{V}}(\bm{A}_{1:n})$ in \eqref{eq:ent}:
\begin{theorem}[Lower bound on exp-entropy]
Suppose $\mathcal{U}$ and $\mathcal{V}$ are fixed. Under \textbf{\textup{(A2)}}, the exp-entropy ${\rm E}_{\mathcal{U},\mathcal{V}}(\bm{A}_{1:n})$ can be lower bounded as:
\small
\begin{equation}
{\rm E}_{\mathcal{U},\mathcal{V}}^{1/n}(\bm{A}_{1:n}) \geq \min_i \left[ {\textup{Var}(y_i)} - \frac{\sigma^2(n-1)}{2} \left\{ \xi_{i,\mathcal{U}}(\bm{R}_{1:n}) + \xi_{i,\mathcal{V}}(\bm{S}_{1:n}) \right\} \right],
\label{eq:lb}
\end{equation}
\normalsize
where $\textup{Var}(y_i) = \sigma^2 \|\mathcal{P}_{\mathcal{U}} \bm{A}_i \mathcal{P}_{\mathcal{V}}\|_F^2 + \eta^2$ (see Lemma \ref{lem:var}), and:
\begin{equation}
\xi_{i,\mathcal{U}}(\bm{R}_{1:n}) := \max_{j:j \neq i} \| (\mathcal{P}_{\mathcal{U}} \bm{R}_i)^T (\mathcal{P}_{\mathcal{U}} \bm{R}_j) \|_2^2.
\label{eq:xi}
\end{equation}
\label{thm:inient}
\end{theorem}


The maximization of the lower bound in \eqref{eq:xi} (which serves as a proxy for ${\rm E}_{\mathcal{U},\mathcal{V}}(\bm{A}_{1:n})$ in \eqref{eq:ent}) can then be connected to the problem of subspace packing. Consider first the variance term $\textup{Var}(y_i) = \sigma^2 \|\mathcal{P}_{\mathcal{U}} \bm{A}_i \mathcal{P}_{\mathcal{V}}\|_F^2 + \eta^2$, which depends on mask $\bm{A}_i$ as well as projection matrices $(\mathcal{P}_{\mathcal{U}},\mathcal{P}_{\mathcal{V}})$. Under \textbf{(A1)} and \textbf{(A2)}, it is easy to see that the expected variance $\mathbb{E}_{\mathcal{P}_{\mathcal{U}},\mathcal{P}_{\mathcal{V}}}\textup{Var}(y_i)$ is constant for any $\bm{A}_i$, $i = 1, \cdots, n$, since uniform priors on $\mathcal{G}_{R,m_1-R}$ and $\mathcal{G}_{R,m_2-R}$ are rotationally invariant. Next, applying \textbf{(A1)} to the terms $\xi_{i,\mathcal{U}}(\bm{R}_{1:n})$ and $\xi_{i,\mathcal{V}}(\bm{S}_{1:n})$ in \eqref{eq:lb}, it follows that a mask design $\bm{A}_{1:n}$ maximizing the lower bound in \eqref{eq:lb} under \textbf{(A1)} and \textbf{(A2)} should have frames $\bm{R}_{1:n}$ and $\bm{S}_{1:n}$ from \eqref{eq:inistr} which jointly minimize:
\begin{equation}
\max_{i \neq j} \|\bm{R}_i^T\bm{R}_j\|_2 \quad \text{and} \quad \max_{i \neq j} \|\bm{S}_i^T\bm{S}_j\|_2.
\end{equation}
\noindent In other words, given no prior information on $\bm{X}$, the initial masks $\bm{A}_{1:n}$ which \textit{maximize} observation entropy should have \textit{low} worst-case block coherences for its frames.

This suggests the following initial mask construction:
\begin{equation}
\bm{A}_i = \bm{R}_i^* \boldsymbol{\Lambda}_i (\bm{S}_i^*)^T = R^{-1/2} \bm{R}_i^* (\bm{S}_i^*)^T, \quad i = 1, \cdots, n,
\label{eq:iniform}
\end{equation}
where $\bm{R}_{1:n}^* = [\bm{R}_1^* \; \cdots \; \bm{R}_n^*]$ and $\bm{S}_{1:n}^* = [\bm{S}_1^* \; \cdots \; \bm{S}_n^*]$ are $R$-frames in $\mathbb{R}^{m_1}$ and $\mathbb{R}^{m_2}$ which minimize their corresponding block coherences. Figure \ref{fig:frpack} visualizes this construction using $n=3$ frames, with matrix dimensions $m_1 = m_2 = 3$ and rank $R=2$. By restricting row and column frames $\bm{R}_{1:n}^*$ and $\bm{S}_{1:n}^*$ to have low block coherence, the row and column spaces for initial masks are well spread-out in terms of their principal angles. These packed frames are then combined one-by-one via matrix multiplication to form initial masks. Given no prior information on $\bm{X}$, the above arguments (from \eqref{eq:lb} and the maximum entropy principle) suggest that initial masks constructed this way (i.e., with \textit{well-packed} row and column frames) can yield \textit{near-maximal information} on $\bm{X}$. We introduce two methods in Section \ref{sec:alg} to construct the well-packed frames $\bm{R}_{1:n}^*$ and $\bm{S}_{1:n}^*$ in \eqref{eq:iniform}, using state-of-the-art algorithms for optimal subspace packings.

\begin{figure}[t]
\centering
\includegraphics[width=0.3\textwidth]{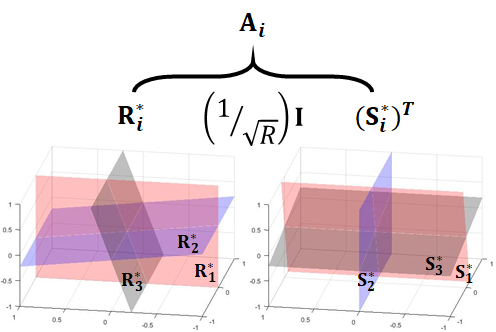}
\caption{A visualization of the mask construction in \eqref{eq:iniform}, highlighting the link between maximum entropy masks and the subspace packing problem.}
\label{fig:frpack}
\vspace{-0.2in}
\end{figure}




\section{Insights on sequential mask design}
\label{sec:seq}
Consider next the case where samples $\bm{y}_n$ have been observed from masks $\bm{A}_{1:n}$, and {suppose a sequence of point estimates $( \hat{\mathcal{U}}_n,\hat{\mathcal{V}}_n )_{n=1, 2, \cdots}$ is obtained on $(\mathcal{U},\mathcal{V})$ from observations $(\bm{y}_n)_{n=1,2,\cdots}$} (more on this in Section \ref{sec:fullalg}). A sequential maximization of the exp-entropy \eqref{eq:ent} yields:
\begin{align}
\begin{split}
\small
&\bm{A}^*_{n+1} := \underset{\bm{A} \in \mathbb{R}^{m_1 \times m_2}, \; \|\bm{A}\|_F^2 \leq 1}{\textup{argmax}} \; {\rm E}_{\hat{\mathcal{U}}_n,\hat{\mathcal{V}}_n}([\bm{A}_{1:n} \; \bm{A}]) \\
& = \underset{\|\bm{A}\|_F^2 \leq 1}{\textup{argmax}} \; \det\{  \sigma^2 \hat{\bm{R}}_{n+1}([\bm{A}_{1:n} \; \; \bm{A}]) + \eta^2 \bm{I}  \}.
\label{eq:seqent}
\end{split}
\end{align}
\normalsize
Here, ${\rm E}_{\hat{\mathcal{U}}_n,\hat{\mathcal{V}}_n}([\bm{A}_{1:n} \; \bm{A}])$ is the joint exp-entropy of observations from masks $\bm{A}_{1:n}$ and $\bm{A}$ with $(\hat{\mathcal{U}}_n, \hat{\mathcal{V}}_n)$ as plug-in estimates for $(\mathcal{U},\mathcal{V})$, and $\hat{\bm{R}}_{n+1}([\bm{A}_{1:n} \; \bm{A}])$ is the correlation matrix in \eqref{eq:rn} with plug-in estimates $(\hat{\mathcal{U}}_n, \hat{\mathcal{V}}_n)$. Equation \eqref{eq:seqent} can be viewed as an information-greedy way to construct measurement masks.


Using the Schur complement \cite{HK1971}, \eqref{eq:seqent} can be further simplified as follows:
\begin{lemma}[Sequential mask optimization]
The optimization in \eqref{eq:seqent} can be rewritten as:
\small
\begin{align}
\begin{split}
\underset{\|\bm{A}\|_F^2 \leq 1}{\textup{argmax}} \; \left\{ \|\mathcal{P}_{\hat{\mathcal{U}}_n} \bm{A} \mathcal{P}_{\hat{\mathcal{V}}_n}\|_F^2 - \hat{\bm{r}}_n^T(\bm{A}) [\hat{\bm{R}}_n(\bm{A}_{1:n}) + \gamma^2 \bm{I}]^{-1} \hat{\bm{r}}_n(\bm{A}) \right\},
\label{eq:seqentform}
\end{split}
\end{align}
\normalsize
where $\hat{\bm{r}}_n(\bm{A})$ is the correlation vector in \eqref{eq:rn} with plug-in estimates $(\hat{\mathcal{U}}_n, \hat{\mathcal{V}}_n)$.
\label{lem:seqent}
\end{lemma}

The simplified problem in \eqref{eq:seqentform} can be interpreted as \textit{subspace matching} in the following sense. To simplify notation, suppose $(\hat{\mathcal{U}}_n,\hat{\mathcal{V}}_n) = (\mathcal{U},\mathcal{V})$. By maximizing the first term $\|\mathcal{P}_{\mathcal{U}} \bm{A} \mathcal{P}_{\mathcal{V}}\|_F^2$, one ensures that the projection of the new mask $\bm{A}$ (onto subspaces of $\bm{X}$) has large norm. This then encourages \textit{subspace matching} of the new mask $\bm{A}$ to the subspaces of $\bm{X}$ (see \cite{Sch1991, DM2008, Bea2010}). On the other hand, by minimizing the second term $\bm{r}_n^T(\bm{A}) [\bm{R}_n(\bm{A}_{1:n}) + \gamma^2 \bm{I}]^{-1} \bm{r}_n(\bm{A})$, we force the new mask $\bm{A}$ to investigate \textit{different} subspaces, previously unexplored by observed masks $\bm{A}_{1:n}$. In this sense, the sequential criterion \eqref{eq:seqentform} offers a balance between subspace exploration and exploitation for mask design. {A similar exploration-exploitation trade-off also arises in the context of active learning for multi-arm bandit problems \cite{Aea2002,Aea2012,Hea2012}.}

\begin{figure}[t]
\centering
\includegraphics[width=0.3\textwidth]{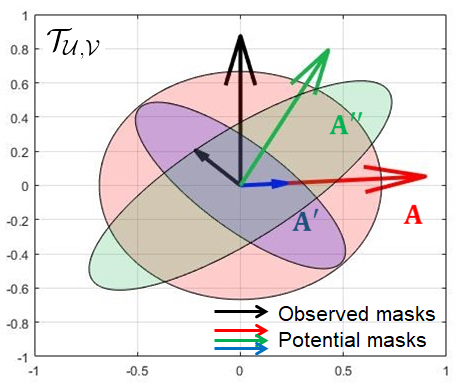}
\caption{Visualizing three potential new masks $\bm{A}$, $\bm{A}'$ and $\bm{A}''$ (red, blue and green vectors) and their covariance matrices, given two observed masks (black vectors).}
\label{fig:seqdet}
\vspace{-0.2in}
\end{figure}

The problem in \eqref{eq:seqentform} also enjoys a nice visualization using the earlier maximum entropy illustration (see Figure \ref{fig:det}). Again, suppose $(\hat{\mathcal{U}}_n,\hat{\mathcal{V}}_n) = (\mathcal{U},\mathcal{V})$ for simplicity. Recall that maximum entropy masks can be viewed as vectors which maximize the volume of their covariance matrix, after projection onto $\mathcal{T}_{\mathcal{U},\mathcal{V}}$. Take now the black and blue vectors in Figure \ref{fig:det} as initial masks; one then aims to select the next mask which maximizes the volume of the covariance matrix ellipse. Figure \ref{fig:seqdet} shows these two initial masks using black vectors, along with the projection of three potential mask choices $\bm{A}$, $\bm{A}'$ and $\bm{A}''$ onto $\mathcal{T}_{\mathcal{U},\mathcal{V}}$, using red, blue and green vectors. Comparing potential masks $\bm{A}$ (red) and $\bm{A}'$ (blue), we see that both $\bm{A}$ and $\bm{A}'$ have the same correlation with observed masks, but $\bm{A}$ has greater projected length. The resulting covariance matrix volume is therefore larger for $\bm{A}$ than for $\bm{A}'$, meaning $\bm{A}$ is a better sequential mask for recovering $\bm{X}$. Likewise, comparing $\bm{A}$ (red) and $\bm{A}''$ (green), we see that both masks have the same projected length, but $\bm{A}$ has smaller correlations with observed masks. The covariance matrix volume is larger for $\bm{A}$ than for $\bm{A}''$, meaning $\bm{A}$ is again a better sequential mask for recovering $\bm{X}$. In this sense, mask $\bm{A}$ satisfies the two-fold objective imposed by the two terms in \eqref{eq:seqentform}.


{The sequential problem in \eqref{eq:seqentform} can also be viewed as a \textit{generalization} of principal components analysis (PCA) for mask design. To see this, first assume the true subspaces $(\mathcal{P}_{\mathcal{U}},\mathcal{P}_{\mathcal{V}})$ are known and fixed, and consider the ``PCA-like'' strategy for mask design:
\begin{equation}
\bm{A}_{n+1}' \leftarrow \argmax_{\substack{||\bm{A}||_F \leq 1,\\
\text{Cov}(y_{\bm{A}}, y_i) = 0, \; \forall i = 1, \cdots, n}} \text{Var}(y_{\bm{A}}|\mathcal{P}_{\mathcal{U}},\mathcal{P}_{\mathcal{V}}), \; n = 1, 2, \cdots,
\label{eq:pca}
\end{equation}
where $y_{\bm{A}}$ is a sample from mask $\bm{A}$. It is easy to see that the vectorized masks from \eqref{eq:pca} can be obtained via the principal components of $\bm{X}$ (i.e., the eigenvectors of the covariance matrix $\text{Var}(\bm{X})$). This PCA approach to mask design has two key restrictions: (a) the true subspaces $(\mathcal{P}_{\mathcal{U}},\mathcal{P}_{\mathcal{V}})$ must be \textit{known} and \textit{fixed} for the entire sequential procedure, and (b) all prior masks must be \textit{uncorrelated} given subspaces $(\mathcal{P}_{\mathcal{U}},\mathcal{P}_{\mathcal{V}})$. In practice, both $\mathcal{P}_{\mathcal{U}}$ and $\mathcal{P}_{\mathcal{V}}$ are unknown, and one needs to rely on adaptive point estimates from data. As the estimates $(\mathcal{P}_{\hat{\mathcal{U}}_n},\mathcal{P}_{\hat{\mathcal{V}}_n})$ change, both restrictions are violated, and so this PCA approach cannot be applied for adaptive mask design.

The formulation in \eqref{eq:seqentform} can be seen as a way to generalize the PCA approach \eqref{eq:pca} for \textit{active} matrix recovery. Using Lemma \ref{thm:conddist}, \eqref{eq:seqentform} can be rewritten as:
\begin{equation}
\bm{A}_{n+1}^* \leftarrow \argmax_{||\bm{A}||_F \leq 1} \; \text{Var}(y_{\bm{A}}|\bm{y}_n, \mathcal{P}_{\hat{\mathcal{U}}_n},\mathcal{P}_{\hat{\mathcal{V}}_n}).
\label{eq:genpca}
\end{equation}
In other words, the information-greedy sequential design \eqref{eq:seqentform} chooses the mask which maximizes the \textit{conditional} variance of an observation, given data $\bm{y}_n$ and current subspace estimates $(\hat{\mathcal{U}}_n,\hat{\mathcal{V}}_n)$. In the special case where (a) $(\mathcal{P}_{\hat{\mathcal{U}}_n}, \mathcal{P}_{\hat{\mathcal{V}}_n})$ are the true subspaces $(\mathcal{P}_{\mathcal{U}}, \mathcal{P}_{\mathcal{V}})$, and (b) all prior masks $\bm{A}_{1:n}$ are uncorrelated, \eqref{eq:genpca} reduces to the PCA construction \eqref{eq:pca}. Outside of this, \eqref{eq:genpca} offers a \textit{generalization} of \eqref{eq:pca} with two key advantages for active matrix recovery. First, by relaxing the uncorrelated mask constraint in \eqref{eq:pca}, \eqref{eq:genpca} permits a sequence of point estimates $( \mathcal{P}_{\hat{\mathcal{U}}_n},\mathcal{P}_{\hat{\mathcal{V}}_n} )_{n=1, 2, \cdots}$ to be plugged-in for mask construction. This then allows for \textit{adaptive} estimation of $(\mathcal{P}_{\mathcal{U}},\mathcal{P}_{\mathcal{V}})$, and thereby \textit{active} matrix recovery. The same cannot be done using the PCA construction, because the subspaces $(\mathcal{P}_{\mathcal{U}},\mathcal{P}_{\mathcal{V}})$ are neither known nor fixed. Second, as we show in Theorem \ref{thm:seqentopt}, \eqref{eq:genpca} yields a nice closed-form solution, which can be used for efficient and adaptive mask design in low-rank problems. Further details on this in Section \ref{sec:seqc}.}


{As more data are collected, the sequence of point estimates $(\hat{\mathcal{U}}_n,\hat{\mathcal{V}}_n)_{n=1,2, \cdots}$ should provide increasingly accurate estimates for the true subspaces $(\mathcal{U},\mathcal{V})$. Our strategy for active mask design is to iterate the following two steps: (a) given data $\bm{y}_n$, compute the nuclear-norm estimate $\hat{\bm{X}}$ from \eqref{eq:nuc}, and obtain subspace estimates $(\hat{\mathcal{U}}_n,\hat{\mathcal{V}}_n)$ via an SVD on $\hat{\bm{X}}$, then (b) plug-in these estimates into \eqref{eq:seqentform} to construct the next mask. From Lemma \ref{thm:mle}, this can be viewed as an MAP-guided active learning scheme on $\bm{X}$. More on this in Section \ref{sec:fullalg}.}

\section{\texttt{MaxEnt} -- Constructing maximum entropy masks}
\label{sec:alg}
We now employ these insights to derive an efficient algorithm \texttt{MaxEnt} for constructing initial and adaptive masks.

\subsection{Initial mask construction}
\label{sec:inic}
Consider first the \textit{initial} design problem of masks $\bm{A}_{1:n}$ in the form of \eqref{eq:iniform} (see Section \ref{sec:ini}). We extend here two subspace packing algorithms from \cite{Cea2015} for constructing the low block coherence frames $\bm{R}^*_{1:n}$ and $\bm{S}^*_{1:n}$ in \eqref{eq:iniform}. The first method, the \textit{flipping construction} (\texttt{ini.flip}), can be used for all matrix dimensions $m_1 \times m_2$ and any initial guess of matrix rank $R_{ini}$. The second method, the \textit{Kerdock-Kronecker construction} (\texttt{ini.kk}), provides higher-quality masks than \texttt{ini.flip}, but can only be used for specific matrix dimensions $m_1 \times m_2$, initial rank $R_{ini}$ and initial sample size $n$. 

\subsubsection{Flipping construction}

\begin{algorithm}[t]
\caption{\texttt{flip}($\bm{R}_{1:n}$) -- Flipping algorithm}
\label{alg:flip}
\setlength{\leftmargini}{10pt}
\bi
\itemsep0em 
\item $\bm{R}_1^* \leftarrow \bm{R}_1$, $\bm{F}_1 \leftarrow \bm{R}_1$
\item \textbf{For} $k = 1, \cdots, n-1$:
\setlength{\leftmargini}{10pt}
\bi
\itemsep0em 
\item \textbf{if} $\|\bm{F}_k + \bm{R}_{k+1}\|_2 \leq \|\bm{F}_k - \bm{R}_{k+1}\|_2$\\
\quad \textbf{then} $\bm{R}_{k+1}^* \leftarrow \bm{R}_{k+1}$\\
\quad \textbf{else} $\bm{R}_{k+1}^* \leftarrow -\bm{R}_{k+1}$
\item $\bm{F}_{k+1}= \bm{F}_k + \bm{R}_{k+1}^*$
\ei
\item \textbf{return} $\bm{R}_{1:n}^* = [\bm{R}_1^* \; \bm{R}_2^* \cdots \bm{R}_n^*]$
\ei
\end{algorithm}

\begin{algorithm}[t]
\caption{\texttt{ini.flip}($m_1,m_2,n,R_{ini}$) -- Flipping construction of initial masks}
\label{alg:flipcon}
\setlength{\leftmargini}{10pt}
\bi
\itemsep0em 
\item Generate uniformly random frames $\bm{R}_{1:n}$ and $\bm{S}_{1:n}$
\item $\bm{R}_{1:n}^* \leftarrow \text{\texttt{flip}}(\bm{R}_{1:n})$, $\bm{S}_{1:n}^* \leftarrow \text{\texttt{flip}}(\bm{S}_{1:n})$
\item $\bm{A}_{1:n} \leftarrow [\bm{A}_1 \; \bm{A}_2 \cdots \bm{A}_n]$, where $\bm{A}_i \leftarrow R_{ini}^{-1/2 }\bm{R}_i^* (\bm{S}_i^*)^T$
\item \textbf{return} $\bm{A}_{1:n}$
\ei
\end{algorithm}

The flipping construction \texttt{ini.flip} relies on the flipping algorithm \texttt{flip} (Algorithm 1 in \cite{Cea2015}) to generate the row and column frames $\bm{R}_{1:n}^*$ and $\bm{S}_{1:n}^*$ in \eqref{eq:iniform}. Given a set of frames $\bm{R}_{1:n}$, \texttt{flip} can be seen as a post-processing step which reduces average block coherence, while retaining low worst-case block coherence. This is achieved by iteratively performing random flips of each frame in $\bm{R}_{1:n}$, and accept such flips only if it results in a reduction in average coherence. The proposed flipping construction then incorporates the frames into \texttt{flip} into \eqref{eq:iniform} to form the initial masks $\bm{A}_{1:n}$. Algorithms \ref{alg:flip} and \ref{alg:flipcon} outlines the details behind \texttt{flip} and \texttt{ini.flip}.

The following lemma from \cite{Cea2015} provides an upper bound guarantee for average block coherence of frames from \texttt{flip}:
\begin{lemma}[Avg. block coherence of \texttt{ini.flip}; Lemma 3.4 in \cite{Cea2015}]
The row and column frames $\bm{R}_{1:n}^*$ and $\bm{S}_{1:n}^*$ returned by \texttt{\textup{flip}} satisfy $a(\bm{R}_{1:n}^*) = a(\bm{S}_{1:n}^*) = (\sqrt{n}+1)/(n-1)$.
\label{lem:flip}
\end{lemma}
\noindent This average coherence upper bound provides an improvement over uniform random frames (see Section 4 of \cite{Cea2015}). Given the link between information and block coherence in Section \ref{sec:inicoh} (with average coherence a proxy for worst-case coherence), this lemma shows that masks constructed using \texttt{ini.flip} yield more initial information on $\bm{X}$ than random masks.

One advantage of \texttt{ini.flip} over \texttt{ini.kk} (introduced next) is that it can be used for \textit{any} matrix dimension or initial rank. However, when $m_1$, $m_2$ and $R_{ini}$ satisfy certain conditions, \texttt{ini.kk} can offer better recovery performance.

\subsubsection{Kerdock-Kronecker (KK) construction}
\label{sec:kkcon}
\begin{algorithm}[t]
\caption{\texttt{ini.kk}($m_1,m_2,n, R_{ini}$) -- Kerdock-Kronecker construction of initial masks}
\label{alg:kkcon}
\setlength{\leftmargini}{10pt}
\bi
\itemsep0em 
\item Generate Kerdock frames for $\bm{K}_{\bm{R}} \in \mathbb{R}^{m_1/R_{ini} \times n}$ and $\bm{K}_{\bm{S}} \in \mathbb{R}^{m_2/R_{ini} \times n}$
\item Generate uniformly random unitary matrices $\bm{Q}_{\bm{R}} \in \mathbb{R}^{R_{ini} \times R_{ini}}$ and $\bm{Q}_{\bm{S}}\in \mathbb{R}^{R_{ini} \times R_{ini}}$
\item $\bm{R}_{1:n}^* \leftarrow \bm{K}_{\bm{R}} \otimes \bm{Q}_{\bm{R}}$, $\bm{S}_{1:n}^* \leftarrow \bm{K}_{\bm{S}} \otimes \bm{Q}_{\bm{S}}$
\item $\bm{R}_{1:n}^* \leftarrow \text{\texttt{flip}}(\bm{R}_{1:n}^*)$, $\bm{S}_{1:n}^* \leftarrow \text{\texttt{flip}}(\bm{S}_{1:n}^*)$
\item $\bm{A}_{1:n} \leftarrow [\bm{A}_1 \; \bm{A}_2 \cdots \bm{A}_n]$, where $\bm{A}_i \leftarrow R_{ini}^{-1/2 }\bm{R}_i^* (\bm{S}_i^*)^T$
\item \textbf{return} $\bm{A}_{1:n}$
\ei
\end{algorithm}

The Kerdock-Kronecker (KK) construction, \texttt{ini.kk}, again uses the form in \eqref{eq:iniform}, but with $\bm{R}_{1:n}^*$ and $\bm{S}_{1:n}^*$ following the Kronecker form:
\begin{equation}
\bm{R}_{1:n}^* = \bm{K}_{\bm{R}} \otimes \bm{Q}_{\bm{R}}, \quad \bm{S}_{1:n}^* = \bm{K}_{\bm{S}} \otimes \bm{Q}_{\bm{S}}.
\label{eq:kron}
\end{equation}
Here, $\bm{K}_{\bm{R}} \in \mathbb{R}^{(m_1/R_{ini}) \times n}$ and $\bm{K}_{\bm{S}} \in \mathbb{R}^{(m_2/R_{ini}) \times n}$ are taken from the Kerdock family of frames \cite{Mea1977}, and $\bm{Q}_{\bm{R}}$, $\bm{Q}_{\bm{S}} \in \mathbb{R}^{R_{ini} \times R_{ini}}$ are independent and uniformly distributed unitary matrices. The following lemma shows that frames constructed this way enjoy low block coherence, both in the average-case and worst-case:
\begin{lemma}[Worst-case and avg. block coherence of \texttt{ini.kk}; Thm. 2.10 and Table 1, \cite{Cea2015}]
The frames $\bm{R}_{1:n}^*$ and $\bm{S}_{1:n}^*$ in \eqref{eq:kron} satisfy (a) $\mu(\bm{R}_{1:n}^*) = 1/\sqrt{m_1}$ and $\mu(\bm{S}_{1:n}^*) = 1/\sqrt{m_2}$, and (b) $a(\bm{R}_{1:n}^*) = a(\bm{S}_{1:n}^*) = 1/(n-1)$.
\label{prop:KK}
\end{lemma}
\noindent The worst-case block coherences in Lemma \ref{prop:KK} nearly achieve the universal coherence lower bound (Theorem 3.6, \cite{LS1973}). Tying this back to Section \ref{sec:ini}, this suggests the initial masks from \texttt{ini.kk} are near-optimal for extracting initial information from $\bm{X}$. In our implementation of \texttt{ini.kk}, \texttt{flip} is used as a post-processing step to further improve average block coherence. Algorithm \ref{alg:kkcon} summarizes the steps for \texttt{ini.kk}.

One restriction of \texttt{ini.kk} is that the Kerdock frames $\bm{K}_{\bm{R}}$ and $\bm{K}_{\bm{S}}$ can be generated only for certain choices of $m_1$, $m_2$, $R_{ini}$ and $n$ (specific conditions in Table \ref{tbl:sum}; see \cite{Hea1994} for details). Despite this, {our construction heuristic via Kerdock codes} has two advantages. First, \texttt{ini.kk} offers improved block coherence (and hence better information-theoretic masks) to \texttt{ini.flip}, and should be used whenever the requirements in Table \ref{tbl:sum} are satisfied (e.g., in image recovery; see Section \ref{sec:ex}). Second, Kerdock frames enjoy a nice packing property, which allows more blocks to be packed into a frame while retaining low block coherence \cite{Cea2015}. This allows designers the option of increasing initial sample size $n$ by setting initial rank guess $R_{ini}$ as small as possible, while ensuring good recovery for matrices with higher rank.

\begin{table}[t]
\caption{Restrictions on matrix dimensions $m_1 \times m_2$ and initial sample size $n$ for \texttt{ini.flip} and \texttt{ini.kk}.}
\centering
\begin{tabular}{c | c c }
\toprule
& \texttt{ini.flip} & \texttt{ini.kk}\\
\hline
$m_1$ & Any & $(2^{k_1 + 1})R_{ini}$, for some odd integer $k_1$\\
$m_2$ & Any & $(2^{k_2 + 1})R_{ini}$, for some odd integer $k_2$\\
$n$ & Any & $n \leq \min(2^{2k_1+2}, 2^{2k_2+2})$\\
\toprule
\end{tabular}
\label{tbl:sum}
\vspace{-0.2in}
\end{table}

\subsection{Sequential mask construction}
\label{sec:seqc}

Consider next the problem of constructing the sequential mask $\bm{A}_{n+1}$ from observed masks $\bm{A}_{1:n}$ (see Section \ref{sec:seq}). Using current subspace estimates $(\hat{\mathcal{U}}_n,\hat{\mathcal{V}}_n)$ (which are adaptively updated via nuclear-norm min.), the following theorem gives a closed form for the sequential optimization in \eqref{eq:seqentform}:
\begin{theorem}[Sequential mask construction]
The optimal \textup{sequential mask} $\bm{A}_{n+1}^*$ in \eqref{eq:seqentform} takes the form:
\begin{equation}
\bm{A}_{n+1}^* = \bm{U} \boldsymbol{\Sigma}^* \bm{V}^T,
\label{eq:seqentopt}
\end{equation}
for any $\bm{U} \in \hat{\mathcal{U}}_n$, $\bm{U}^T\bm{U} = \bm{I}$ and $\bm{V} \in \hat{\mathcal{V}}_n$, $\bm{V}^T\bm{V} = \bm{I}$. Here, $\textup{vec}(\boldsymbol{\Sigma}^*) \in \mathbb{R}^{R^2}$ is the unit eigenvector for the smallest eigenvalue of $\bm{D}^T [\hat{\bm{R}}_n(\bm{A}_{1:n}) + \gamma^2 \bm{I}]^{-1} \bm{D}$, $\bm{D} = [\textup{vec}(\boldsymbol{\Sigma}_1)^T; \cdots ; \textup{vec}(\boldsymbol{\Sigma}_n)^T] \in \mathbb{R}^{n \times R^2}$, with $\boldsymbol{\Sigma}_i = \bm{U}^T\bm{A}_i\bm{V}$.
\label{thm:seqentopt}
\end{theorem}
\noindent This theorem can be explained as follows. Having sampled from masks $\bm{A}_{1:n}$, the optimal sequential mask is constructed using the row and column frames $\bm{U}$ and $\bm{V}$ (from estimated subspaces $\hat{\mathcal{U}}_n$ and $\hat{\mathcal{V}}_n$), combined using a weight matrix $\boldsymbol{\Sigma}^*$. This weight matrix, obtained via a minimum eigenvector computation, can be seen as an \textit{optimal allocation} of measurement power to the row and column frames $\bm{U}$ and $\bm{V}$, to minimize correlations between the new mask $\bm{A}_{n+1}^*$ and observed masks $\bm{A}_{1:n}$. Equivalently, the power allocation in $\boldsymbol{\Sigma}^*$ can be viewed as distributing measurement power to important subspaces yet to be explored by previous masks $\bm{A}_{1:n}$.

Computationally, the key appeal of Theorem \ref{thm:seqentopt} is that it provides a \textit{closed-form} mask construction for greedy information gain on $\bm{X}$. Compared to a numerical optimization of \eqref{eq:seqentform}, which is computationally infeasible even for moderate-dim. problems, the closed-form solution in \eqref{eq:seqentopt} offers a much more efficient construction of sequential masks. This closed form follows from the maximum entropy principle, which allows us to work with the simpler observation entropy as a proxy for more complicated entropy or error criteria (see \eqref{eq:entmmse}).

\begin{figure}[t]
\centering
\includegraphics[width=0.27\textwidth]{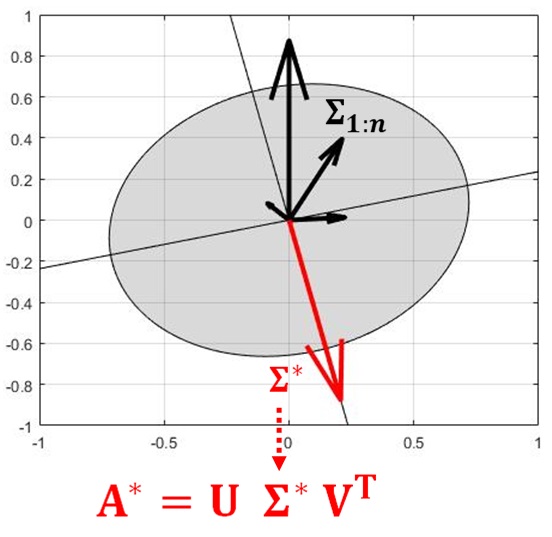}
\caption{The sequential mask construction in \eqref{eq:seqentopt}, visualized as an optimal power allocation on subspaces.}
\label{fig:ev}
\end{figure}

The solution in \eqref{eq:seqentopt} also offers a geometric interpretation. For fixed row and column frames $\bm{U}$ and $\bm{V}$, consider the representation of observed masks $\{\bm{A}_i\}_{i=1}^n$ by its power allocation matrices $\{\boldsymbol{\Sigma}_i\}_{i=1}^n$, where $\boldsymbol{\Sigma}_i = \bm{U}^T \bm{A}_i\bm{V}$ (as in Theorem \ref{thm:seqentopt}). Note that $\langle \mathcal{P}_{\mathcal{U}} \bm{A}_i \mathcal{P}_{\mathcal{V}}, \mathcal{P}_{\mathcal{U}} \bm{A}_j \mathcal{P}_{\mathcal{V}} \rangle_F = \langle \boldsymbol{\Sigma}_i, \boldsymbol{\Sigma}_j \rangle_F$, so, similar to Section \ref{sec:maxent}, the goal of maximum entropy masks can be viewed as maximizing the covariance matrix \textit{volume} for vectors of the observed power matrices $\{\boldsymbol{\Sigma}_i\}_{i=1}^n$. Figure \ref{fig:ev} visualizes these vectors and their covariance ellipse for $n=4$ observed masks. The minimum eigenvector solution for $\boldsymbol{\Sigma}^*$ in \eqref{eq:seqentopt} can be viewed as the \textit{minor axis} -- the axis with shortest length -- of the covariance ellipse. This is quite intuitive, because adding the vector $\boldsymbol{\Sigma}^*$ (red arrow in Figure \ref{fig:ev}) along the minor axis maximizes the \textit{volume gain} of the ellipse. The sequential mask yielding the greatest \textit{information gain} is then constructed by using the weight matrix $\boldsymbol{\Sigma}^*$ to optimally allocate \textit{measurement power} to the row and column frames $\bm{U}$ and $\bm{V}$ (bottom of Figure \ref{fig:ev}). This interpretation is related to the ``water-filling'' allocation in information-theoretic frame design for compressive sensing (see \cite{Cea2012,CT2012}), except instead of allocating more measurement power to \textit{less noisy} channels, the optimal mask in \eqref{eq:seqentopt} allocates more power to frames (a) \textit{more correlated} with the desired matrix $\bm{X}$ and (b) \textit{less correlated} with previous masks.

\subsection{Full algorithm}
\label{sec:fullalg}

\begin{algorithm}[t]
\caption{\texttt{MaxEnt}($m_1,m_2,n_{ini},R_{ini},n_{seq}$)  -- Information-theoretic matrix recovery}
\label{alg:maxent}
\setlength{\leftmargini}{10pt}
\bi
\itemsep0em 
\item \textbf{If} conditions in Table \ref{tbl:sum} satisfied\\ 
\textbf{then} $\bm{A}_{1:n_{ini}} \leftarrow \texttt{ini.kk}(m_1, m_2, n_{ini}, R_{ini})$\\
\textbf{else} $\bm{A}_{1:n_{ini}} \leftarrow \texttt{ini.flip}(m_1, m_2, n_{ini}, R_{ini})$
\item Observe initial samples $y_i \leftarrow \langle \bm{A}_i,\bm{X} \rangle_F + \epsilon_i$, $i = 1, \cdots, n_{ini}$
\item \textbf{For} $i = n_{ini}, \cdots, n_{ini} + n_{seq}-1$:
\bi
\itemsep0em 
\item Estimate $\hat{\bm{X}}_i$ via the nuclear-norm formulation \eqref{eq:nuc}
\item Estimate row and column spaces $(\hat{\mathcal{U}}_i,\hat{\mathcal{V}}_i)$ from $\texttt{svd}(\hat{\bm{X}}_i)$
\item Construct next mask $\bm{A}_{i+1}$ from \eqref{eq:seqentopt} using $(\hat{\mathcal{U}}_i,\hat{\mathcal{V}}_i)$
\item Observe new sample $y_{i+1} \leftarrow \langle \bm{A}_{i+1},\bm{X} \rangle_F + \epsilon_{i+1}$
\ei
\item \textbf{Return} recovered matrix $\hat{\bm{X}}_{n_{ini}+n_{seq}}$
\ei
\end{algorithm}
Algorithm \ref{alg:maxent} summarizes the full \texttt{MaxEnt} procedure for information-theoretic matrix recovery, which consists of three steps. First, initial masks $\bm{A}_{1:n_{ini}}$ are generated using either \texttt{ini.kk} (whenever possible) or \texttt{ini.flip}. Next, the nuclear-norm solution $\hat{\bm{X}}$ in \eqref{eq:nuc} is estimated from observed data; this serves as a close approximation to the MAP estimator (Lemma \ref{thm:mle}). In our implementation, $\hat{\bm{X}}$ is optimized via the Matlab solver CVX \cite{Gea2008} for small matrices, and a straight-forward extension of the singular-value thresholding algorithm \cite{Cea2010} for larger matrices. An SVD on $\hat{\bm{X}}$ then yields point estimates for the row and column spaces $\hat{\mathcal{U}}$ and $\hat{\mathcal{V}}$. Lastly, using these subspace estimates, the next mask is then constructed using the closed-form equation \eqref{eq:seqentopt}. The last two steps (subspace estimation and active sampling) are then repeated until a desired sample size is obtained, or a desired error tolerance achieved. For larger problems, this iterative procedure can be sped up using batch sampling, by supplementing the sequential construction \eqref{eq:seqentopt} with (a) masks constructed using the smallest few eigenvectors from the power allocation in \eqref{eq:seqentopt}, or (b) randomly sampled masks.

{Of course, in the above iterative procedure, the adaptive MAP estimates $(\hat{\mathcal{U}}_n,\hat{\mathcal{V}}_n)$ may deviate from the true subspaces $(\mathcal{U}, \mathcal{V})$, particularly early on in sampling. Indeed, one disadvantage with MAP-guided active sampling is that it fails to account for parameter uncertainty from estimation error. From a Bayesian perspective, this uncertainty can be incorporated via a \textit{fully-Bayesian} approach, which replaces the objective function in the sequential optimization \eqref{eq:seqentform} with its \textit{expectation} under the posterior distribution $[\mathcal{P}_{\mathcal{U}},\mathcal{P}_{\mathcal{V}}|\bm{y}_n]$ (which quantifies uncertainty in subspaces $(\mathcal{U},\mathcal{V})$ after observing data). However, this fully-Bayesian design problem is computationally infeasible to optimize for two reasons. First, the \textit{evaluation} of the expected objective requires simulation from the non-standard posterior distribution $[\mathcal{P}_{\mathcal{U}},\mathcal{P}_{\mathcal{V}}|\bm{y}]$, which can be very costly to sample (see \cite{MX2017}). Second, the nice \textit{closed-form} solution for sequential masks from Theorem \ref{thm:seqentopt} (for fixed $\mathcal{P}_{\mathcal{U}}$ and $\mathcal{P}_{\mathcal{V}}$) is no longer available for the expected objective. Because of this, this fully-Bayesian formulation is very time-consuming to optimize, even for small matrices $\bm{X}$. We find that the proposed MAP-guided approach in \texttt{MaxEnt} offers a more efficient adaptive strategy for learning $\bm{X}$.}

\vspace{-0.1in}
\section{Numerical examples}
\label{sec:num}

\begin{figure}[!tp]
\centering
\includegraphics[width=0.43\textwidth]{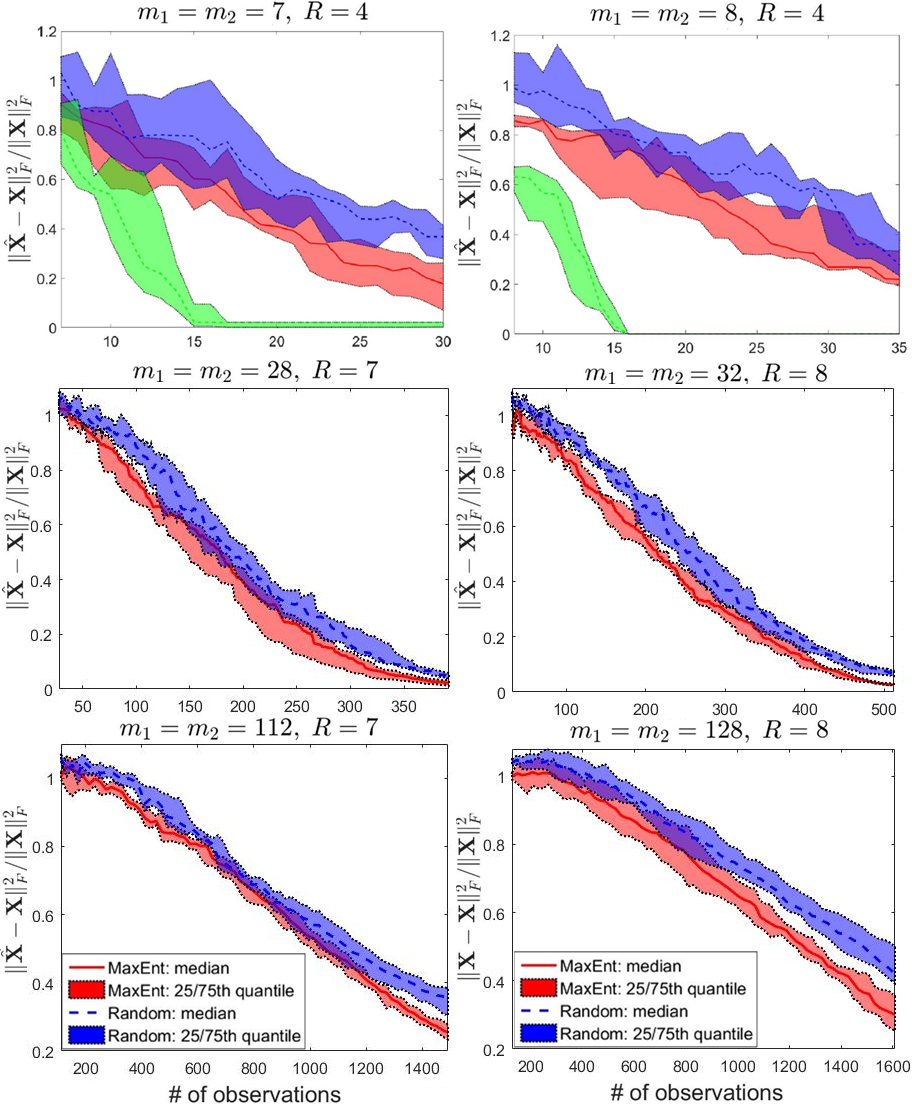}
\vspace{-0.2cm}
\caption{Normalized recovery errors for simulated $\bm{X} \in \mathbb{R}^{m_1 \times m_2}$ with rank $R$, using \texttt{\textup{MaxEnt}} (left: \texttt{\textup{ini.flip}}, right: \texttt{\textup{ini.kk}}) and random masks. Solid lines mark median error, and shaded bands mark 25-th/75-th error quantiles. {Errors from the PCA construction \eqref{eq:pca} are shown in green; these are optimal masks when true subspaces are known.}}
\label{fig:sim}
\vspace{-0.2in}
\end{figure}

\subsection{Simulated examples}
\label{sec:sim}
We now investigate the performance of \texttt{MaxEnt} for recovering simulated instances of $\bm{X}$. Here, $\bm{X}$ is simulated using the SMG model $\mathcal{SMG}(\mathcal{P}_{\mathcal{U}}, \mathcal{P}_{\mathcal{V}},\sigma^2,R)$, with uniformly sampled $\mathcal{P}_{\mathcal{U}}$ and $\mathcal{P}_{\mathcal{V}}$ and variance parameters $\sigma^2 = 1$ and $\eta^2 = 10^{-4}$. Six simulation cases are conducted for different matrix dimensions and rank $(m_1, m_2, R)$: the first three cases $(7,7,4)$, $(28,28,7)$ and $(112,112,\cbl{7})$ investigate \texttt{MaxEnt} using the flipping construction \texttt{ini.flip}, while the next three cases $(2^3,2^3,4)$, $(2^5,2^5,8)$ and $(2^7,2^7,\cbl{8})$ investigate the Kerdock-Kronecker construction \texttt{ini.kk} ($R_{ini}$ is set as 2 to exploit the packing property of Kerdock frames; see Section \ref{sec:kkcon}). The first three cases employ an initial and total sample size $(n_{ini},n_{ini}+n_{seq})$ of $(7,30)$, $\cbl{(28,400)}$ and $(112,\cbl{1500})$, while the next three cases use $(8,35)$, $\cbl{(32,500)}$ and $(128,\cbl{1600})$ samples. For each case, we replicate the simulation for ten trials to provide an estimate of error variability. 

For each of the six cases, Figure \ref{fig:sim} shows the normalized recovery errors $\|\hat{\bm{X}} - \bm{X}\|_F^2 / \|\bm{X}\|_F^2$ as a function of sample size, where $\hat{\bm{X}}$ is the nuclear-norm estimate \eqref{eq:nuc}. To benchmark performance, the proposed method \texttt{MaxEnt} is compared with uniform random masks satisfying the unit power constraints. Consider first the \textit{initial} recovery performance of \texttt{MaxEnt} (using \texttt{ini.flip} or \texttt{ini.kk}) and random masks. For the three cases on the left, the initial masks from \texttt{ini.flip} give noticeable improvements over random masks, both for median error and error quantiles. For the three cases on the right, the initial masks from \texttt{ini.kk} yield more pronounced improvements over random masks; the 75-th error quantiles for the former are smaller than the 25-th error quantiles for the latter. These results corroborate two earlier insights. First, the improvement of \texttt{ini.flip} and \texttt{ini.kk} over random masks supports the link between information and block coherence of frames (see Section \ref{sec:inicoh}). Second, this confirms the improved performance of \texttt{ini.kk} over \texttt{ini.flip}, which is expected since the former has better packing properties than the latter (see Section \ref{sec:inic}).

Consider next the \textit{sequential} performance of \texttt{MaxEnt}. From Figure \ref{fig:sim}, the sequential recovery from \texttt{MaxEnt} is noticeably better than random masks, at nearly all sample sizes. This error gap appears to grow larger with more sequential samples, which suggests the closed-form adaptive design in Section \ref{sec:seqc} is indeed effective. By designing masks which greedily maximize information on $\bm{X}$, \texttt{MaxEnt} provides an informed sampling scheme which adaptively targets \textit{important} subspaces of $\bm{X}$. This growing error gap also hints at an improved theoretical rate for \texttt{MaxEnt} over random masks, which we leave for future work.

{Lastly, for the two smaller cases (Figure \ref{fig:sim}, top), we compare \texttt{MaxEnt} with the PCA masks from \eqref{eq:pca} -- the latter can be viewed as an \textit{oracle} design scheme when the \textit{true subspaces} of $\bm{X}$ are \textit{known with certainty}, prior to data. Not surprisingly, given access to the true subspaces $(\mathcal{U},\mathcal{V})$, this PCA approach yields improved recovery to \texttt{MaxEnt}, achieving near-perfect recovery after $R^2=16$ samples. The error gap between this approach and \texttt{MaxEnt} can be attributed to the extra samples needed to adaptively learn the underlying subspaces. We note that this PCA mask construction is \textit{not} viable for practical problems, since one rarely knows (with certainty) the true subspaces of an \textit{unknown} matrix $\bm{X}$ prior to data. }

\subsection{Real data examples}
\label{sec:ex}

\begin{figure}[t]
\centering
\includegraphics[width=0.45\textwidth]{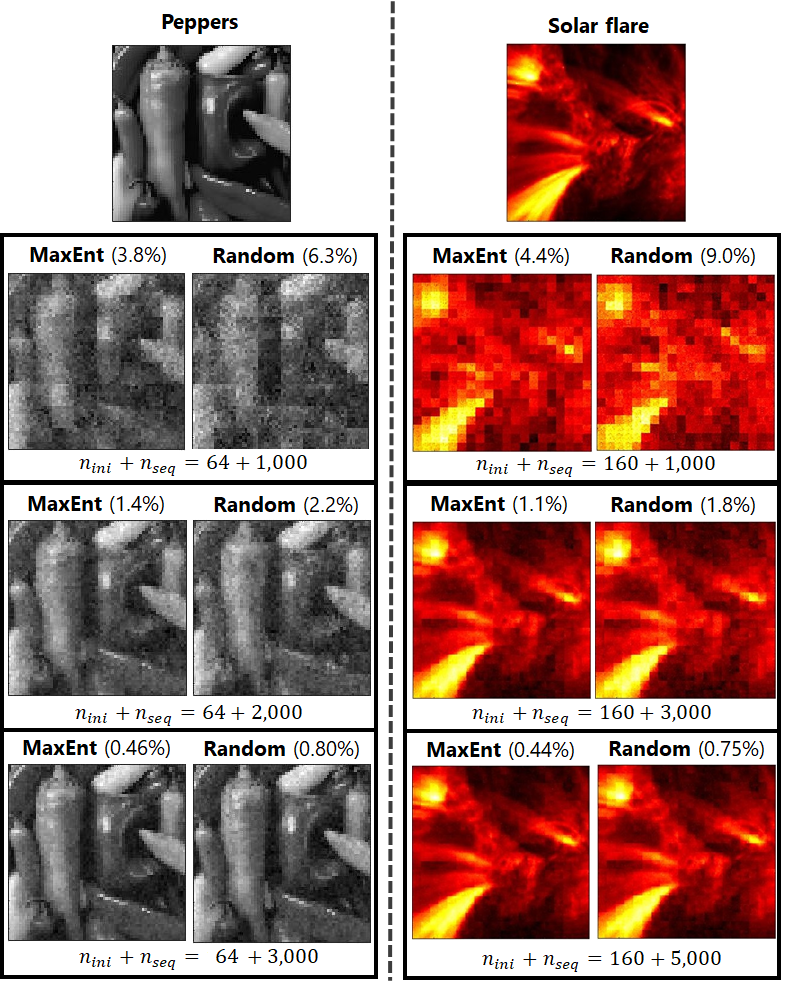}
\vspace{-0.1in}
\caption{(Top) The original `peppers' and `flare' images. (Bottom) The recovered images using \texttt{\textup{MaxEnt}} and random masks, with $n_{ini}+n_{seq}$ measurements. Normalized errors are bracketed.}
\label{fig:img}
\vspace{-0.2in}
\end{figure}

\begin{figure}[t]
\centering
\includegraphics[width=0.5\textwidth]{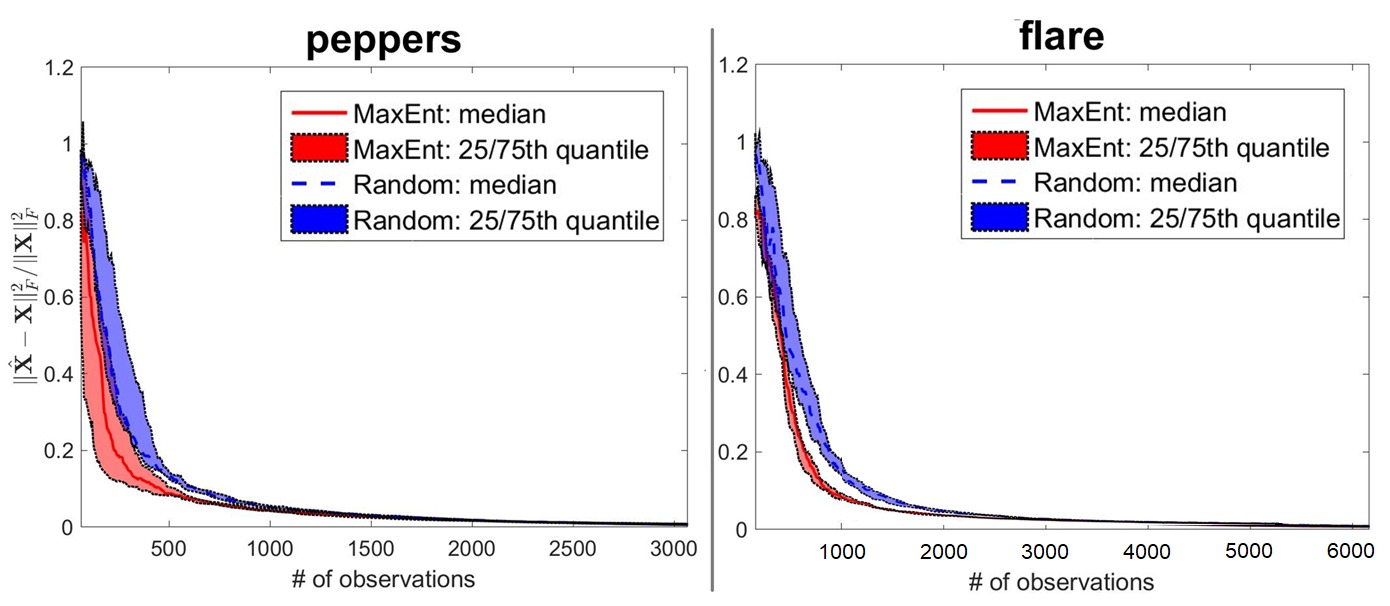}
\vspace{-0.2in}
\caption{Normalized recovery errors for `peppers' and `flare' using \texttt{\textup{MaxEnt}} and random masks. Solid lines mark median errors, and shaded bands mark 25-th/75-th error quantiles.}
\label{fig:imgerr}
\vspace{-0.2in}
\end{figure}

\subsubsection{Image recovery}

Next, we investigate the performance of \texttt{MaxEnt} for recovering (a) `peppers' -- a $64 \times 64$-pixel peppers image\footnote{\scriptsize \ULurl{www.statemaster.com/encyclopedia/Standard-test-image}.}, and (b) `flare' -- a {$160 \times 160$}-pixel solar flare image captured by the NASA SDO satellite (see \cite{Xea2013} for details). These images are shown in Figure \ref{fig:img} (top). To generate $\bm{X}$, we first break each image into $8 \times 8$ patches, then vectorize these patches and collect vectors into an $m_1 \times m_2 = 64 \times 64$ matrix for `peppers', and an {$m_1 \times m_2 = 64 \times 400$} matrix for `flare' (details on this patching step in \cite{CX2016}). Using this patched matrix $\bm{X}$, observations are then sampled with $\eta^2 = 1.0$. For `peppers', $n_{ini}=64$ initial masks are constructed using \texttt{ini.kk} with $R_{ini}=4$, with $n_{seq} = 3,000-64$ samples taken sequentially; for `flare', $n_{ini} = 160$ initial masks are constructed using \texttt{ini.flip} with $R_{ini} = 8$, with $n_{seq} = 3,000 - 160$ samples taken sequentially. This is replicated five times to give an estimate of error variability.

For these images, Figure \ref{fig:imgerr} shows the normalized recovery errors $\|\hat{\bm{X}} - \bm{X}\|_F^2 / \|\bm{X}\|_F^2$ as a function of measurements taken. As before, \texttt{MaxEnt} is compared with uniformly random masks. For \textit{initial} recovery, \texttt{MaxEnt} yields markedly lower errors to random sampling for both images, which again shows the effectiveness of well-packed subspaces for maximizing initial information gain. As before, \texttt{ini.kk} provides greater error reduction to \texttt{ini.flip}. For \textit{sequential} recovery, \texttt{MaxEnt} maintains a sizable error gap over random sampling for both images, particularly early on in the sequential procedure. This illustrates the ability of \texttt{MaxEnt} to \textit{learn} and \textit{target} important image features via adaptive mask design. These results are in line with the observations in Section \ref{sec:sim}.

For a visual comparison, Figure \ref{fig:img} (bottom) shows the original images for `peppers' and `flare', and the recovered images using \texttt{MaxEnt} and random masks for different sample sizes. For `peppers', \texttt{MaxEnt} provides noticeably improved recovery of key image characteristics, such as the distinct shape and lighting of each individual pepper, whereas the same characteristics appear more blurred for random masks. Likewise, for `flare', \texttt{MaxEnt} yields a clearer recovery of key solar flare features, such as the intensity and spray of each flare eruption, whereas the same features appear more blurred for random masks. This nicely visualizes the ability of \texttt{MaxEnt} to actively learn important image features via an adaptive, information-theoretic mask design.

\subsubsection{Text document indexing}
\label{sec:data}
We now explore the usefulness of \texttt{MaxEnt} for compressing (or indexing) large text databases into smaller, representative datasets. The database to compress, $\bm{X} \in \mathbb{R}^{m_1 \times m_2}$, is compiled from police reports provided by the Atlanta Police Department \cite{ZX2018}, on crimes committed between the years 2014 -- 2017. Each row of $\bm{X}$ corresponds to a separate police report (with $m_1=497$ report narratives in total), and each column of $\bm{X}$ records the frequency of a specific term in a bag-of-words representation \cite{MS1999} of these reports (with $m_2=100$ terms of interest). For \texttt{MaxEnt}, an initial design of $n_{ini} = 1,000$ masks is constructed using \texttt{ini.flip} with $R_{ini}=8$, with $n_{seq} = 9,000$ sequential samples, resulting in a compressed dataset of $n_{tot} = 10,000$ samples. This is then compared with (a) random masks ($n_{tot}=10,000$ samples), and (b) {a hybrid approach with random initial masks ($n_{ini} = 1,000$) followed by sequential \texttt{MaxEnt} sampling ($n_{seq}=9,000$)}. All methods are compared on how well the reduced dataset recovers the original police report database. 

\begin{table}[t]
\caption{Normalized recovery errors $\|\hat{\bm{X}} - \bm{X}\|_F^2/\|\bm{X}\|_F^2$ for text document indexing, using random masks, \texttt{MaxEnt}, and a hybrid approach (random initialization with sequential \texttt{MaxEnt}).}
\centering
\small
\begin{tabular}{c | c c }
\toprule
 & $n_{ini} = 1,000$ & $n_{tot} = 10,000$\\
\hline
\textbf{Random} & 91.6\% & 17.6\% \\
\textbf{\texttt{MaxEnt}} & 88.4\% & 4.1\% \\
{\textbf{Hybrid}} & {91.6\%} & {5.4\% }\\
\toprule
\end{tabular}
\label{tbl:apd}
\vspace{-0.2in}
\end{table}

Table \ref{tbl:apd} summarizes the normalized recovery errors $\|\hat{\bm{X}} - \bm{X}\|_F^2/\|\bm{X}\|_F^2$ for \texttt{MaxEnt} and random masks, where $\hat{\bm{X}}$ is the nuclear-norm-recovered database from the sketched dataset. Consider first the compression using $n_{ini} =1,000$ initial samples. \texttt{ini.flip} offers lower recovery errors to random masks, which demonstrates the importance of well-packed frames for initial compression. Consider next the compression using $n_{ini}+n_{seq}=10,000$ total samples. Here, the full \texttt{MaxEnt} procedure provides a growing error improvement over random masks, which reflects the \textit{learning} and \textit{targeting} of important subspace properties in $\bm{X}$ for adaptive sampling. Indeed, in latent semantic analysis (see, e.g., \cite{MS1999}), the SVD of the document-term matrix $\bm{X}$ encodes important information on (a) different document types found in the database, and (b) typical terms found within each document type. Viewed this way, the adaptive procedure in \texttt{MaxEnt} first \textit{learns} this latent document-term structure, then \textit{exploits} such structure to design effective masks for compression.

{Comparing next the final compression errors for the three methods, we see that the hybrid approach offers significant improvements over random masks, and is only slightly worse than the full \texttt{MaxEnt} approach. This shows that, at least for text document indexing, the adaptive design aspect of \texttt{MaxEnt} may be more important than initial mask design. Such a conclusion is not too surprising, since we know there exists some document-term structure within the text data; by learning this structure and targeting it via active sampling, the proposed adaptive approach in \texttt{MaxEnt} can offer much improved compression over random sampling.

Lastly, on computation cost, the running time for \texttt{MaxEnt} for this example is 2,591 sec. (a little more than half an hour), on a single-core 3.4 Ghz processor, whereas random sampling requires only several seconds for mask construction. Interestingly, only a small part of the \texttt{MaxEnt} time (129 sec.) is spent on computing the closed-form solution \eqref{eq:seqentopt}; the rest is spent on solving the nuclear-norm problem \eqref{eq:nuc} for subspace learning. Given the increasing availability of multicore and parallel processing, the latter step can be greatly sped up using \textit{distributed} matrix recovery algorithms (see, e.g., \cite{Mea2015}) -- an active topic in machine learning. These distributed algorithms can allow \texttt{MaxEnt} to be comparable with random sampling not only on running time, but also on problem scalability (currently, \texttt{MaxEnt} on a single-threaded processor can efficiently handle matrices as large as $750 \times 750$). We look forward to exploring this in a future work.}

\vspace{-0.1in}
\section{Conclusion}
In this paper, we proposed a novel information-theoretic approach for designing measurement masks for low-rank matrix recovery. We first revealed novel insights on the link between mask design and {subspace packings}. Using such insights, we then developed an algorithm (\texttt{MaxEnt}) for efficiently constructing initial and adaptive measurement masks to maximize information on $\bm{X}$.

Looking forward, there are several interesting directions to pursue next. First, while the numerical results in Section \ref{sec:num} show a considerable advantage of \texttt{MaxEnt} compared to random masks, it would be nice to quantify this via a theoretical rate. Second, given the promising applications to image processing and text document indexing here, we are interested in exploring the usefulness of \texttt{MaxEnt} in other low-rank modeling problems in engineering and statistics, including covariance sketching \cite{Cea2015}, system identification \cite{LV2009}, physics extraction \cite{Mea2017}, and gene association studies \cite{ND2014, MW2017}. {It would also be nice to develop a fully-Bayesian implementation of \texttt{MaxEnt}, which can then be used to quantify uncertainty on both $\bm{X}$ and its subspaces. Finally, further study of the exploitation-exploration trade-off in \eqref{eq:seqentform} is also of interest.}

%

\vspace{-0.1in}
\bibliography{references}

\newpage
\appendices
\section{Proofs}

\begin{proof}[Proof of Lemma \ref{thm:mle}]
The proof follows from a direct application of Lemma 2 in \cite{MX2017}.
\end{proof}

\begin{proof}[Proof of Lemma \ref{lem:var}]
Let $\bm{x}_1, \cdots, \bm{x}_{m_2} \in \mathbb{R}^{m_1}$ denote the $m_2$ columns in $\bm{X}$, and let $\bm{a}_{i,1}, \cdots, \bm{a}_{i,m_2} \in \mathbb{R}^{m_1}$ denote the $m_2$ columns in $\bm{A}_i$. Note that:
\begin{align*}
\text{Var}(y_i) &= \text{Cov}\{\text{tr}(\bm{A}_i^T \bm{X}) + \epsilon_i, \text{tr}(\bm{A}_i^T \bm{X}) + \epsilon_i \}\\
&= \text{Cov}\left\{ \sum_{k_1=1}^{m_2} \bm{a}_{i,k_1}^T \bm{x}_{k_1}, \sum_{k_2=1}^{m_2} \bm{a}_{i,k_2}^T \bm{x}_{k_2}\right\} + \eta^2\\
& = \sum_{k_1=1}^{m_2} \sum_{k_2=1}^{m_2} \bm{a}_{i,k_1}^T \text{Cov} (\bm{x}_{k_1}, \bm{x}_{k_2})\bm{a}_{i,k_2} + \eta^2\\
& = \sigma^2 \sum_{k_1=1}^{m_2} \sum_{k_2=1}^{m_2} [\mathcal{P}_{\mathcal{V}}]_{k_1,k_2} ( \bm{a}_{i,k_1}^T \mathcal{P}_{\mathcal{U}} \bm{a}_{i,k_2}) + \eta^2\\
&= \sigma^2 \; \text{tr}\{ (\mathcal{P}^{1/2}_{\mathcal{U}} \bm{A}_i) \mathcal{P}_{\mathcal{V}}(\mathcal{P}^{1/2}_{\mathcal{U}} \bm{A}_i)^T \} + \eta^2\\
&= \sigma^2 \| \mathcal{P}_{\mathcal{U}} \bm{A}_i \mathcal{P}_{\mathcal{V}}\|_F^2 + \eta^2.
\end{align*}
Using analogous steps, it follows that $\text{Cov}(y_i,y_j) = \sigma^2 \; \text{tr}(\mathcal{P}_{\mathcal{U}} \bm{A}_i \mathcal{P}_{\mathcal{V}} \bm{A}_j^T) = \sigma^2 \langle \mathcal{P}_{\mathcal{U}} \bm{A}_i \mathcal{P}_{\mathcal{V}}, \mathcal{P}_{\mathcal{U}} \bm{A}_j \mathcal{P}_{\mathcal{V}} \rangle_F$ for $i \neq j$, which completes the proof.
\end{proof}

\begin{proof}[Proof of Lemma \ref{thm:conddist}]
The proof follows from a straight-forward extension of Lemma \ref{lem:var}, and the closed-form conditional distribution of a multivariate Gaussian random vector (see, e.g., Theorem 3.3.4 in \cite{Ton2012}).
\end{proof}

\begin{proof}[Proof of Theorem \ref{thm:inient}]
\noindent The proof of this theorem requires the following lemmas:
\begin{lemma}[Gershgorin's Circle Theorem; \cite{Ger1931}]
Let $\bm{R} \in \mathbb{C}^{n \times n}$, and let $\mathcal{B}_i := \mathcal{B}(R_{i,i}, r_i)$ be the closed ball in the complex plane with center $R_{i,i}$ and radius $r_i = \sum_{j:j \neq i} |R_{i,j}|$. Then all the eigenvalues of $\bm{R}$ lie in the union $\cup_{i=1}^n \mathcal{B}_i$.
\label{lem:gersh}
\end{lemma}

\begin{lemma}[Upper bound on inner-product]
Assume the initial masks $\bm{A}_{1:n}$ follow $\textup{\textbf{(A2)}}$. Then, for fixed $i$ and $j$:
\begin{align}
\begin{split}
& \Big| \Big\langle \mathcal{P}_{\mathcal{U}} \bm{A}_i \mathcal{P}_{\mathcal{V}}, \mathcal{P}_{\mathcal{U}} \bm{A}_j \mathcal{P}_{\mathcal{V}} \Big\rangle_F \Big| \leq \\
& \quad \frac{1}{2} \left\{ \| (\mathcal{P}_{\mathcal{U}} \bm{R}_i)^T (\mathcal{P}_{\mathcal{U}} \bm{R}_j) \|_2^2 + \| (\mathcal{P}_{\mathcal{V}} \bm{S}_i)^T (\mathcal{P}_{\mathcal{V}} \bm{S}_j) \|_2^2 \right\}.
\end{split}
\end{align}
\normalsize
\label{lem:inprod}
\end{lemma}

\begin{proof}[Proof of Lemma \ref{lem:inprod}]
Under \textbf{(A2)}, we have $\bm{A}_i = R^{-1/2} \bm{R}_i \bm{S}_i^T$ for $i = 1, \cdots, n$. The inner-product term then becomes:
\begin{align*}
& \Big| \Big\langle \mathcal{P}_{\mathcal{U}} \bm{A}_i \mathcal{P}_{\mathcal{V}}, \mathcal{P}_{\mathcal{U}} \bm{A}_j \mathcal{P}_{\mathcal{V}} \Big\rangle_F \Big| \\
& \quad = \frac{1}{R} \Big| {\rm tr}\left\{ (\mathcal{P}_{\mathcal{U}} \bm{R}_i)^T (\mathcal{P}_{\mathcal{U}} \bm{R}_j) (\mathcal{P}_{\mathcal{V}} \bm{S}_j)^T (\mathcal{P}_{\mathcal{V}} \bm{S}_i) \right\} \Big| \\
& \quad \leq \frac{1}{2R} \left\{ \| (\mathcal{P}_{\mathcal{U}} \bm{R}_i)^T (\mathcal{P}_{\mathcal{U}} \bm{R}_j) \|_F^2 + \| (\mathcal{P}_{\mathcal{V}} \bm{S}_i)^T (\mathcal{P}_{\mathcal{V}} \bm{S}_j) \|_F^2 \right\}\\
& \quad \leq \frac{1}{2} \left\{ \| (\mathcal{P}_{\mathcal{U}} \bm{R}_i)^T (\mathcal{P}_{\mathcal{U}} \bm{R}_j) \|_2^2 + \| (\mathcal{P}_{\mathcal{V}} \bm{S}_i)^T (\mathcal{P}_{\mathcal{V}} \bm{S}_j) \|_2^2 \right\}.
\end{align*}
\end{proof}

Consider now the matrix $\sigma^2\bm{R}_n(\bm{A}_{1:n}) + \eta^2 \bm{I}$, which is symmetric and positive-definite. By Lemma \ref{lem:gersh}, the eigenvalues of this matrix (which are real-valued and positive) can be lower bounded by $\min_i \{ \text{Var}(y_i) - \sum_{j:j \neq i}|\text{Cov}(y_i,y_j)| \}$. Viewing the determinant as a product of eigenvalues, it follows that:
\begin{align*}
\small
&{\rm E}^{1/n}_{\mathcal{U},\mathcal{V}}(\bm{A}_{1:n})\\
&\quad = {\det}^{1/n}\{\sigma^2 \bm{R}_n(\bm{A}_{1:n}) + \eta^2 \bm{I} \} \\
&\quad \geq \min_i \left\{ {\text{Var}(y_i)} - \sum_{j:j \neq i} {|\text{Cov}(y_i,y_j)|} \right\} \tag{Lemma \ref{lem:gersh}} \\
&\quad = \min_i \left\{ {\text{Var}(y_i)} - \sigma^2 \sum_{j:j \neq i} \Big| {\Big\langle \mathcal{P}_{\mathcal{U}} \bm{A}_i \mathcal{P}_{\mathcal{V}}, \mathcal{P}_{\mathcal{U}} \bm{A}_j \mathcal{P}_{\mathcal{V}} \Big\rangle_F \Big| } \right\} \tag{Lemma \ref{lem:var}}\\
&\quad \geq \min_i \Big\{ {\text{Var}(y_i)} - \frac{\sigma^2}{2} \sum_{j:j \neq i} \left\{ \| (\mathcal{P}_{\mathcal{U}} \bm{R}_i)^T (\mathcal{P}_{\mathcal{U}} \bm{R}_j) \|_2^2 \right. \\
& \left. \quad \quad \quad \quad \quad \quad \quad \quad \quad \quad \quad \quad \quad \quad  + \| (\mathcal{P}_{\mathcal{V}} \bm{S}_i)^T (\mathcal{P}_{\mathcal{V}} \bm{S}_j) \|_2^2 \right\} \Big\} \tag{Lemma \ref{lem:inprod}}\\
&\quad \geq \min_i \left\{ {\textup{Var}(y_i)} - \frac{\sigma^2(n-1)}{2} \left\{ \xi_{i,\mathcal{U}}(\bm{R}_{1:n}) + \xi_{i,\mathcal{V}}(\bm{S}_{1:n}) \right\} \right\},
\end{align*}
\normalsize
which completes the proof.
\end{proof}


\begin{proof}[Proof of Lemma \ref{lem:seqent}]
For notational brevity, let $(\mathcal{U},\mathcal{V}) = (\hat{\mathcal{U}}_n,\hat{\mathcal{V}}_n)$. First, write the matrix $\bm{R}_{n+1}([\bm{A}_{1:n} \; \; \bm{A}]) + \gamma^2\bm{I}$ as:
\small
\[
\bm{R}_{n+1}([\bm{A}_{1:n} \; \; \bm{A}]) + \gamma^2\bm{I} = \begin{pmatrix}
\bm{R}_{n}(\bm{A}_{1:n}) + \gamma^2\bm{I} & \bm{r}_n(\bm{A})\\
\bm{r}_n(\bm{A}) & \|\mathcal{P}_{\mathcal{U}} \bm{A} \mathcal{P}_{\mathcal{V}}\|_F^2 + \gamma^2
\end{pmatrix}.
\]
\normalsize
Using the determinant formula for the Schur complement (see, e.g., \cite{HK1971}), it follows that:
\begin{align*}
\begin{split}
&{\rm H}(\bm{A}_{1:n}, \bm{A})\\
& \quad = (\sigma^2)^{n+1} \det\{ \bm{R}_{n+1}([\bm{A}_{1:n} \; \; \bm{A}]) + \gamma^2\bm{I}\}\\
& \quad = (\sigma^2)^{n+1} \det\{\bm{R}_{n}(\bm{A}_{1:n}) + \gamma^2\bm{I}\} \cdot \\
& \quad \quad \left[ \|\mathcal{P}_{\mathcal{U}} \bm{A} \mathcal{P}_{\mathcal{V}}\|_F^2 + \gamma^2 - \bm{r}_n^T(\bm{A}) (\bm{R}_{n}(\bm{A}_{1:n}) + \gamma^2\bm{I})^{-1} \bm{r}_n(\bm{A}) \right]\\
& \quad = \sigma^2 \textup{H}(\bm{A}_{1:n}) \cdot \\
& \quad \quad \left[ \|\mathcal{P}_{\mathcal{U}} \bm{A} \mathcal{P}_{\mathcal{V}}\|_F^2 + \gamma^2 - \bm{r}_n^T(\bm{A}) (\bm{R}_{n}(\bm{A}_{1:n}) + \gamma^2\bm{I})^{-1} \bm{r}_n(\bm{A}) \right],
\end{split}
\end{align*}
which completes the proof.
\end{proof}

\begin{proof}[Proof of Theorem \ref{thm:seqentopt}]
For notational brevity, let $(\mathcal{U},\mathcal{V}) = (\hat{\mathcal{U}}_n,\hat{\mathcal{V}}_n)$. The proof of this theorem requires the following lemmas:
\begin{lemma}
For fixed projection matrices $\mathcal{P}_{\mathcal{U}}$ and $\mathcal{P}_{\mathcal{V}}$, define the operator $\mathcal{P}_{\mathcal{U},\mathcal{V}}: \mathbb{R}^{m_1 \times m_2} \rightarrow \mathbb{R}^{m_1 \times m_2}$ as $\mathcal{P}_{\mathcal{U},\mathcal{V}}(\bm{Z})  = \mathcal{P}_{\mathcal{U}}\bm{Z} \mathcal{P}_{\mathcal{V}}$, and consider the linear space of matrices:
\begin{equation}
\mathcal{T}_{\mathcal{U},\mathcal{V}} = \underset{u_k \in \mathcal{U}, v_k \in \mathcal{V} }{\bigcup} \textup{span}(\{\bm{u}_k\bm{v}_k^T\}_{k=1}^R). \label{Tau_def}
\end{equation}
It follows that $\mathcal{P}_{\mathcal{U},\mathcal{V}}$ is an orthogonal projection operator onto $\mathcal{T}_{\mathcal{U},\mathcal{V}}$ under the Frobenius inner-product $\langle \cdot, \cdot \rangle_F$.
\label{lem:oproj}
\end{lemma}
\begin{proof}
This can be shown from first principles. Note that:
\ben
\item $\mathcal{P}_{\mathcal{U},\mathcal{V}}$ is idempotent, i.e., $\mathcal{P}_{\mathcal{U},\mathcal{V}} \{\mathcal{P}_{\mathcal{U},\mathcal{V}} (\bm{Z})\} = \mathcal{P}_{\mathcal{U},\mathcal{V}} (\bm{Z})$ for all $\bm{Z} \in \mathbb{R}^{m_1 \times m_2}$,
\item $\mathcal{P}_{\mathcal{U},\mathcal{V}}$ is the identity operator on $\mathcal{T}_{\mathcal{U},\mathcal{V}}$,
\item $\bm{Z}$ can be uniquely decomposed as $\bm{Z} = \mathcal{P}_{\mathcal{U},\mathcal{V}}(\bm{Z}) + [\mathcal{I} - \mathcal{P}_{\mathcal{U},\mathcal{V}}](\bm{Z})$, where $\mathcal{P}_{\mathcal{U},\mathcal{V}}(\bm{Z}) \in \mathcal{T}_{\mathcal{U},\mathcal{V}}$, $[\mathcal{I} -\mathcal{P}_{\mathcal{U},\mathcal{V}}](\bm{Z}) \in \mathcal{T}^{\bot}_{\mathcal{U},\mathcal{V}}$, and $\bot$ is the orthogonal complement.
\een
By definition, $\mathcal{P}_{\mathcal{U},\mathcal{V}}$ must be a projection operator. Moreover, for any $\bm{Z}_1, \bm{Z}_2 \in \mathbb{R}^{m_1 \times m_2}$, $\langle \mathcal{P}_{\mathcal{U},\mathcal{V}} (\bm{Z}_1), \bm{Z}_2 - \mathcal{P}_{\mathcal{U},\mathcal{V}}(\bm{Z}_2)  \rangle_F = \langle \mathcal{P}_{\mathcal{U},\mathcal{V}} (\bm{Z}_2), \bm{Z}_1 - \mathcal{P}_{\mathcal{U},\mathcal{V}}(\bm{Z}_1)  \rangle_F = 0$,
so $\mathcal{P}_{\mathcal{U},\mathcal{V}}$ must also be an orthogonal projection operator under $\langle \cdot, \cdot \rangle_F$. By Lemma 1(a) of \cite{MX2017}, the range of $\mathcal{P}_{\mathcal{U},\mathcal{V}}$ is $\mathcal{T}_{\mathcal{U},\mathcal{V}}$, which completes the proof.
\end{proof}

\begin{lemma}
Let $f(\bm{A}) := \| \mathcal{P}_{\mathcal{U},\mathcal{V}}(\bm{A}) \|_F^2$, where $\|\bm{A}\|_F^2 = 1$. If $\bm{A} \in \mathcal{T}_{\mathcal{U},\mathcal{V}}$, then $f(\bm{A}) = 1$; otherwise, $f(\bm{A}) < 1$.
\label{lem:1}
\end{lemma}
\begin{proof}
We know from Lemma \ref{lem:oproj} that $\mathcal{P}_{\mathcal{U},\mathcal{V}}$ is an orthogonal projection operator onto $\mathcal{T}_{\mathcal{U},\mathcal{V}}$ under $\langle \cdot, \cdot \rangle_F$. It follows that:
\small
\begin{align*}
1 = \|\bm{A}\|_F^2 &= \| \mathcal{P}_{\mathcal{U},\mathcal{V}}(\bm{A}) + [\mathcal{I} - \mathcal{P}_{\mathcal{U},\mathcal{V}}](\bm{A}) \|_F^2 \\
& = \| \mathcal{P}_{\mathcal{U},\mathcal{V}}(\bm{A})\|_F^2 + \|[\mathcal{I} - \mathcal{P}_{\mathcal{U},\mathcal{V}}](\bm{A}) \|_F^2 \geq \| \mathcal{P}_{\mathcal{U},\mathcal{V}}(\bm{A})\|_F^2,
\end{align*}
\normalsize
with equality holding iff $\|[\mathcal{I} - \mathcal{P}_{\mathcal{U},\mathcal{V}}](\bm{A}) \|_F^2 = 0$, or equivalently, $\bm{A} \in \mathcal{T}_{\mathcal{U},\mathcal{V}}$. This completes the proof.
\end{proof}

Using these two lemmas, the proof for Theorem \ref{thm:seqentopt} is straight-forward. Consider first the maximization of the first term in \eqref{eq:seqentform}, $f(\bm{A}) = \|\mathcal{P}_{\mathcal{U}}\bm{A} \mathcal{P}_{\mathcal{V}}\|_F^2$. By Lemma \ref{lem:1}, $f(\bm{A})$ is maximized at 1 for any choice of $\bm{A} \in \mathcal{T}_{\mathcal{U},\mathcal{V}}$, or equivalently, any $\bm{A}$ of the form $\bm{U} \boldsymbol{\Sigma} \bm{V}^T$, where $\bm{U} \in \mathcal{U}$ and $\bm{V} \in \mathcal{V}$. Consider next the minimization of the second term $g(\bm{A}) := \bm{r}_n^T(\bm{A}) [\bm{R}_n(\bm{A}_{1:n}) + \gamma^2 \bm{I}]^{-1} \bm{r}_n(\bm{A})$. Note that, for any feasible solution $\bm{A}$ satisfying $\|\bm{A}\|_F^2 \leq 1$, the matrix $\widetilde{\bm{A}} = \mathcal{P}_{\mathcal{U}} \bm{A} \mathcal{P}_{\mathcal{V}}$ must also be a feasible solution with the same objective $g(\widetilde{\bm{A}}) = g(\bm{A})$, since $\|\widetilde{\bm{A}}\|_F^2 \leq 1$ by Lemma \ref{lem:1}. Hence, the optimal solution for \eqref{eq:seqentform} must take the form $\bm{A} = \bm{U} \boldsymbol{\Sigma}\bm{V}^T$ for some $\boldsymbol{\Sigma} \in \mathbb{R}^{R \times R}$. Moreover, because $f(\bm{A}) = 1$ for all $\bm{A} = \bm{U} \boldsymbol{\Sigma} \bm{V}^T$, the problem reduces to finding an optimal choice of $\boldsymbol{\Sigma}$ which minimizes the second term $g(\bm{A})$.

With this in mind, $\bm{r}_n(\bm{A})$ can be rewritten as:
\begin{align*}
\bm{r}_n(\bm{A}) &= [\langle \mathcal{P}_{\mathcal{U}} \bm{A}_i \mathcal{P}_{\mathcal{V}}, \mathcal{P}_{\mathcal{U}} \bm{A} \mathcal{P}_{\mathcal{V}} \rangle_F]_{i=1}^n\\
&= [\textup{tr}( \boldsymbol{\Sigma}_i^T \boldsymbol{\Sigma} ) ]_{i=1}^n \tag{$\boldsymbol{\Sigma}_i := \bm{U}^T \bm{A}_i\bm{V}$}\\
&= \left[ \text{vec}(\boldsymbol{\Sigma}_i)^T \text{vec}(\boldsymbol{\Sigma}) \right]_{i=1}^n\\
&= \bm{D}\boldsymbol{\nu},
\end{align*}
where $\bm{D} = [\text{vec}(\boldsymbol{\Sigma}_1)^T; \cdots ; \text{vec}(\boldsymbol{\Sigma}_n)^T] \in \mathbb{R}^{n \times R^2}$ and $\boldsymbol{\nu} = \text{vec}(\boldsymbol{\Sigma}) \in \mathbb{R}^{R^2}$. The desired objective $g(\bm{A})$ can then be rearranged as:
\begin{align*}
g(\bm{A}) &= \bm{r}_n^T(\bm{A}) [\bm{R}_n(\bm{A}_{1:n}) + \gamma^2 \bm{I}]^{-1} \bm{r}_n(\bm{A})\\
&= \boldsymbol{\nu}^T \bm{D}^T [\bm{R}_n(\bm{A}_{1:n}) + \gamma^2 \bm{I}]^{-1} \bm{D} \boldsymbol{\nu},
\end{align*}
so the optimal $\boldsymbol{\Sigma}$ which minimizes $g(\bm{A})$ corresponds to the unit eigenvector for the smallest eigenvalue of $\bm{D}^T [\bm{R}_n(\bm{A}_{1:n}) + \gamma^2 \bm{I}]^{-1} \bm{D}$. Letting $\boldsymbol{\Sigma}^*$ denote this optimal matrix, it follows that $\bm{A}_{n+1}^* = \bm{U} \boldsymbol{\Sigma}^* \bm{V}^T$, which completes the proof.
\end{proof}

%

%

%
%
%




\end{document}